\documentclass[lettersize,journal]{IEEEtran}
\usepackage{amsmath,amssymb,amsfonts}
\usepackage{algorithmic}
\usepackage{graphicx}
\usepackage{textcomp}
\usepackage{stfloats}
\usepackage{xcolor}
\usepackage{subfigure}
\usepackage{amsthm}
\newtheorem{theorem}{Theorem}
\usepackage{comment}
\usepackage{enumitem}
\usepackage{amsmath,amsfonts}

\newtheorem{lemma}{Lemma}

\usepackage{cite}
\begin{document}
\title{ Over-the-Air Interference Nulling Using Passive RIS for Two-Way $K$-User Interference Channel}
\author{
    \IEEEauthorblockN{Junzhi Wang, Jun Sun, Limin Liao, Xiangbai Liao, Yingzhuang Liu}\\
     \thanks{Junzhi Wang, Jun Sun, Limin Liao, and Yingzhuang Liu are with the School of Electronic Information and Communications, Huazhong University of Science and Technology, Wuhan, China.

     Xiangbai Liao is with Hunan Institute of Technology, Hengyang, China.} 
}

\maketitle

\begin{abstract}
Interference constitutes the fundamental performance bottleneck in wireless networks. Meanwhile, reconfigurable intelligent surface (RIS) has emerged as a promising technique for interference mitigation by directly modifying wireless channels. In this paper, we are interested in the following problem: whether \textit{interference-free} transmission (in terms of Degree-of-Freedom, DoF) can be achieved with the aid of passive RIS in the two-way K-user interference channel, which is regarded as the most severely interfered network.
We show that the answer is affirmative, i.e., interference in this network can be neutralized over the air. To accomplish this goal, two prominent challenges arise: i) the unit-modulus constraint on each RIS reflecting coefficient; ii) the significant disparity between the strengths of the direct and reflective channels.
To address these challenges, we exploit the high-dimensional and random nature of wireless channels. Specifically, we cast the problem within a high-dimensional convex geometric framework, which enables us to leverage the ubiquitous \textit{concentration} phenomenon in high-dimensional spaces.
Based on this framework, we establish both sufficient and necessary conditions on the required number of RIS elements to achieve interference-free DoF, which turns out to \textit{coincide} in order sense.
Furthermore, we characterize the impact of imperfect channel state information (CSI) on the achievable DoF and show that interference-free DoF remains achievable if the CSI error is below a certain threshold.
Simulation results validate our theoretical findings.

\end{abstract}


\section{Introduction}

Interference is a fundamental property of wireless channels due to their broadcast and superposition nature. With the unprecedented connection density envisioned for 6G networks, spectrum scarcity is further exacerbated, rendering time–frequency resource sharing inevitable. As a result, efficient interference management becomes critical for achieving the high throughput and reliability required by future 6G systems\cite{ heydarishahreza2024spectrum, torrens2024modeling, tang2024interference, abode2025goal}.

Meanwhile, the reconfigurable intelligent surface (RIS) has emerged as a promising technique for interference management, as it can directly reconfigure wireless channels \cite{bafghi2022degrees,jiang2022interference,chae2022cooperative,10706811}, thereby providing a more practical solution than the celebrated interference alignment (IA) technique \cite{cadambe2008interference}.
Specifically, \cite{jiang2022interference} demonstrates  \textit{numerically} that almost interference nulling can be achieved in the $K$-user interference channel, i.e., interference-free degree of freedom (DoF) are attainable with the aid of a finite-size RIS. Furthermore, \cite{wang2024achieving} proves  \textit{theoretically} that interference-free transmission can indeed be achieved in a class of more general interference networks, namely cellular networks, by casting the problem into the framework of high-dimensional probability and convex geometry.
Motivated by these inspiring results, we aim to investigate the feasibility of interference-free transmission in a more challenging setting, namely the two-way $K$-user interference channel (TW-K-IFC), which is regarded as the most severely interfered network since every node experiences interference from all the other nodes.
Specifically, we aim to address the following two issues:

\begin{enumerate}[label=\arabic*)] 
\item \textit{Is interference-free transmission possible with the aid of passive RIS in TW-K-IFC, considering the unit-modulus constraints of each RIS element (RE) and the huge disparity between the direct and reflected links? If so, how many REs are needed at least?}

Mathematically, achieving interference-free transmission is equivalent to characterizing the feasibility of a random linear system, i.e., ${\bf{Ax}} + {\bf{b}} = 0$, where ${\bf{A}}$ and ${\bf{b}}$ represent the interference channels of the cascaded reflective and direct links, respectively, and ${\bf{x}} \in \mathbb{C}^N$ denotes the adjustment provided by the RIS, with $N$ being the number of REs.
If ${\bf{x}}$ is unconstrained, it is obvious that this system is solvable if and only if ${\text{rank}}\left( {\mathbf{A}} \right) = {\text{rank}}\left( {\left[ {{\mathbf{A}}|{\mathbf{b}}} \right]} \right) \leqslant N$. 
However, for passive RIS, each RE can only apply a phase shift ($\left|x_i\right| = 1, \forall i$), thus the columns of ${\bf A}$ can be combined only via phase rotations. As a result, canceling a strong  direct link ${\bf b}$ requires more columns, while the exact number remains an open problem.

\item \textit{How does imperfect CSI affect the DoF of the TW-IFC?}

The performance of RIS depends on both the number of REs and CSI accuracy. In practice, acquiring accurate CSI is challenging due to the passive nature of RIS. The CSI errors may hinder interference cancellation and degrade system performance. However, this degradation is difficult to characterize because the modulus-1 constraints of passive REs render closed-form expressions for optimal phase shifts intractable, especially in multi-user systems. 

\end{enumerate}


It is worth noting that the two-way IFC is full-duplex in nature, 
thus efficient self-interference cancellation (SIC) is essential. Since this paper focuses on network-level interference management, we assume that SIC can be achieved with negligible performance degradation using the latest techniques \cite{islam2019unified,liao2025analog,liao2025digital}, in which self-interference can be suppressed to the noise floor. Consequently, inter-user interference becomes the dominant limiting factor, and this work focuses on its mitigation.

\subsection{Prior Works}

Compared with the conventional interference alignment, which requires symbol extension and achieves the DoF of $\frac{K}{2}$ in the $K$-user interference channel \cite{cadambe2008interference}, recent studies have shown that RIS can completely eliminate interference without any symbol extension, thereby achieving the interference-free or full DoF $K$ \cite{bafghi2022degrees,jiang2022interference,chae2022cooperative}. 
In particular, \cite{bafghi2022degrees} demonstrated that if the number of passive REs approaches infinity, the full DoF $K$ can be achieved. A tighter bound was provided in \cite{jiang2022interference}, showing that under weak direct channels, slightly more than $2K(K-1)$ passive REs are enough to achieve the full DoF $K$. Furthermore, by simulations, the authors of \cite{jiang2022interference} revealed that when the strength of the direct channel exceeds a certain threshold, additional passive REs are required. However, the precise relationship between direct channel strength and the required number of passive REs for interference-free communication remains an open problem.

Most existing studies assume perfect CSI. In practice, however, acquiring accurate RIS-related CSI is highly challenging due to the lack of radio frequency (RF) chains and signal processing capabilities at the passive RIS \cite{gao2023robust,li2023performance}. As RIS design heavily relies on CSI, unavoidable estimation errors may substantially degrade system performance. For instance, \cite{zhou2020framework} demonstrated that when the cascaded reflective channel error is large, the base station (BS) transmit power must increase with the number of REs ($N$) to meet the quality-of-service requirements. This contrasts with the perfect CSI case, where the transmit power decreases with the number of REs, scaling as $O\left( {{N^2}} \right)$ \cite{wu2019intelligent}.

Moreover, imperfect CSI can also cause substantial capacity degradation \cite{yang2021performance}. To characterize this effect, the authors of \cite{singh2024imperfect} derived an exact mean signal-to-noise ratio (SNR) expression for a single-user system and demonstrated that, as the error variance increases, the growth trend of the mean SNR reduces from $O\left( {{N^2}} \right)$ to $O\left( N \right)$. 
However, deriving such scaling laws for multi-user systems is more challenging due to the severe inter-user interference. 
When the base station (BS) is equipped with $M$ antennas exceeding the number of single-antenna users, \cite{zhi2022ris} revealed that the uplink rate with zero-forcing (ZF) detection scales as $O\left( {{{\log }_2}\left( {MN} \right)} \right)$. In contrast, with maximal ratio combining (MRC) detection, the rate scales only as $O\left( {{{\log }_2}\left( 1 \right)} \right)$ with respect to $N$ \cite{zhi2022analysis,zhi2022two}. 
To optimize semantic computation rates under jamming and CSI imperfection in a multi-functional RIS- and mobile-edge-computing-assisted network, \cite{sun2024multi} transformed the jammer’s CSI imperfection into a worst-case one via discretization. This effective approach for imperfect CSI has also been adopted in multi-high-altitude-platform-assisted networks \cite{lin2025secure} and rate-splitting-assisted stacked metasurface transceivers \cite{sun2025dual} to maximize system capacity.
Nevertheless, the sum-rate (or DoF) characterization of single-antenna systems under severe inter-user interference with imperfect CSI remains an open problem.

\subsection{Main Contributions}
To tackle the aforementioned challenges for achieving interference-free transmission in the TW-K-IFC, we follow the idea in \cite{wang2024achieving} and \cite{wang2025interference}, which leverages the high dimensionality and randomness of wireless channel coefficients. Specifically, the feasibility problem is cast into the framework of high-dimensional probability and convex geometry, allowing us to exploit the universal \textit{concentration} phenomenon in high-dimensional spaces. Based on this approach, we establish that interference-free DoF in the TW-K-IFC is indeed achievable with the aid of RIS. Furthermore, we derive the necessary and sufficient conditions on the number of REs for achieving this goal. The key results of this paper are summarized as follows:

\begin{enumerate}[label=\arabic*)]
\item 
For the scenario with perfect CSI, the necessary and sufficient conditions on the number of REs for interference-free DoF (i.e., $2K$) are of the same order, thereby defining the critical point (${N_{\text{crit}}}$) for the feasibility of interference-free transmission. Specifically,
\begin{equation}
    {N_{{\text{crit}}}} = \left\{ {\begin{aligned}
&O\left( {{K^2}} \right),\;{\text{if}}\:\eta < K\\
&O\left( {K\eta } \right),\;{\text{if}}\:\eta > {K}
\end{aligned}} \right.
,
\end{equation}
where $K$ is the number of user pairs, $\eta$ denotes the link-strength ratio\footnote{Specifically, $\eta  = \sqrt {\eta _1^2 + \eta _2^2}$. For the necessary condition, $\eta_1$ and $\eta_2$ correspond to the ratios of the weakest same-side and cross-side direct links to the strongest cascaded reflective link, respectively. In contrast, for the sufficient condition, they correspond to the ratios of the strongest direct links on each side to the weakest cascaded reflective link.}  of the direct link to cascaded link.
Intuitively, when the direct links are weak relative to the cascaded reflective links, the required number of REs depends only on the number of user pairs; otherwise, it additionally depends on the link-strength ratio,  implying more REs are required.

\item
For the scenario with imperfect CSI, a sum DoF of $2K\min \left( {\alpha,1} \right)$ is achievable, where $\alpha$ denotes the CSI error exponent, i.e., the CSI error variance scales as $\text{SNR}^{-\alpha}$. This holds when the number of REs satisfies
\begin{equation}
N \geqslant {\bar N_\text{suff}} = \left\{ \begin{gathered}
  O\left( {{K^2}} \right),\;{\text{if}}\:\bar\eta < K, \hfill \\
  O\left( {K\bar\eta } \right),\;{\text{if}}\:\bar\eta > K, \hfill \\ 
\end{gathered}  \right.
\end{equation}
where $\bar\eta$ denotes the estimated link-strength ratio\footnote{Here $\bar\eta  = \sqrt {\bar\eta _1^2 + \bar\eta _2^2}$, and the strength ratios ${\bar\eta _1} = \frac{{{\bar\sigma _2}}}{\varsigma{{\underline{\sigma}_1}}}$ and ${\bar\eta _2} = \frac{{{\bar\sigma _3}}}{{\varsigma{\underline{\sigma}_1}}}$.  
${{\bar \sigma }_2}$ and ${{\bar \sigma }_3}$ represent the path loss of the strongest same-side and cross-side direct links, respectively.  $\underline{\sigma}_1$ represents the weakest cascaded reflective link. 
The CSI accuracy factor $\varsigma  = \sqrt {1 - {\text{SN}}{{\text{R}}^{ - \alpha }}}$.}.
Clearly, if $\alpha$ is no smaller than 1, there would be no DoF loss. Fortunately, in practice, the common channel estimators (LS, MMSE) achieve $\alpha=1$.

\end{enumerate}

\subsection{Paper Organization and Notation}
The rest of the paper is organized as follows. Section~\ref{System_Model} introduces the system model. Section~\ref{PRE} analyzes the required number of passive REs for interference-free transmission under perfect CSI. Section~\ref{Law} investigates the system DoF under imperfect CSI. Section~\ref{Simulation} presents the numerical results. Finally, Section~\ref{Conclusion} concludes the paper.

\emph{Notations:} Scalars, vectors, and matrices are denoted by lowercase ($x$), boldface lowercase (${\bf{x}}$), and boldface uppercase letters (${\bf{X}}$), respectively. The operators ${\left(  \cdot  \right)^*}$, ${\left(  \cdot  \right)^T}$, and ${\left(  \cdot  \right)^H}$ denote the conjugate, transpose, and conjugate transpose, respectively. ${\left(  \cdot  \right)^ + }$ and $\left|  \cdot  \right|$ represent the pseudo-inverse and Frobenius norm. $E\left(  \cdot  \right)$ and $\operatorname{var} \left(  \cdot  \right)$ denote the expectation and variance. ${\text{diag}}\left( {\mathbf{x}} \right)$ returns a diagonal matrix with ${\mathbf{x}}$ on the diagonal, and ${{\mathbf{I}}_K}$ denotes the $K \times K$ identity matrix. ${\text{pro}}{{\text{j}}_{\mathbf{A}}}\left( {\mathbf{x}} \right)$ denotes the projection of ${\mathbf{x}}$ onto the column space of ${\mathbf{A}}$.

\section{System Model}
\label{System_Model}
The RIS-aided two-way $K$-user interference channel is illustrated in Fig. \ref{fig1}, where a passive RIS with $N$ reflective elements is deployed between the $K$ single-antenna user pairs. Since all users operate in full-duplex mode, they can transmit and receive signals simultaneously. Consequently, each user $k$ receives not only the desired signal from its pair $k+K$, but also the interference from other users on both the same and cross sides within the same time slot.

For theoretical insights, we first analyze the ideal case with perfect CSI for both the RIS and all users \cite{peng2021multiuser,chen2023next}, under which the required number of passive REs for interference-free transmission is derived. We then extend the analysis to the imperfect CSI case and characterize the resulting system DoF in the latter part of the paper.

\begin{figure}
	\centering\includegraphics[width=5.5cm]{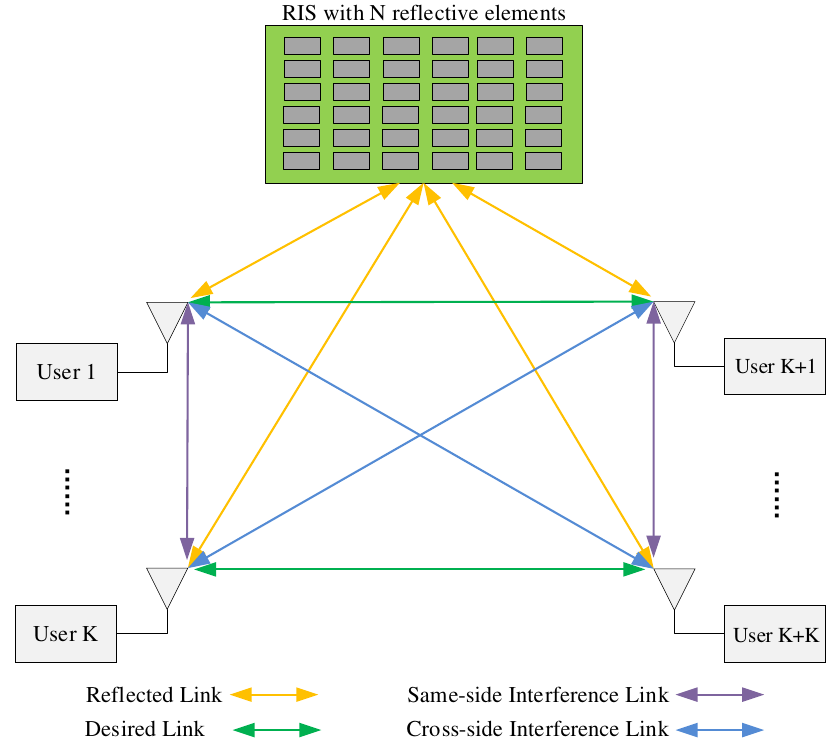}
	\caption{The RIS-aided two-way $K$-user interference channel.}
	\label{fig1}
\end{figure}

\subsection{Channel Model}
Let ${{\mathbf{h}}_{Ii}} \in {\mathbb{C}^{N \times 1}}$ denote  the channel from user $i$ to the RIS, and 
${{\mathbf{h}}_{iI}}\in {\mathbb{C}^{N \times 1}}$ denote  the channel from the RIS to user $i$. 
The direct channel from user $j$ to user $i$ is denoted by ${h_{ij}} \in {\mathbb{C}^{1 \times 1}}$.
A time-division duplexing protocol is adopted, where all transmissions operate at the same carrier frequency but in different time slots within the channel coherence interval. Hence, the physical propagation channels are assumed to be reciprocal for all links \cite{atapattu2020reconfigurable,ma2021joint}.

Following \cite{jiang2022interference}, which demonstrates that the phase transition in the required number of REs (i.e., $N$) for interference nulling is independent of the channel model
, we model all channels as Rayleigh fading for simplicity.
Moreover, although path loss varies with user locations, the primary factor affecting $N$ is the intensity disparity between the direct and cascaded reflective links. To characterize how this intensity disparity affects $N$, we assume uniform path loss\footnote{Assuming uniform path loss is both practical and convenient for analysis. Since stronger direct links cause more severe interference, more REs are required for suppression. Hence, the strongest direct link with the weakest reflective link represents the worst case for determining the maximum number of REs needed to guarantee interference-free transmission, while the weakest direct link with the strongest reflective link determines the minimum number of REs. } 
for each link type. 
Specifically, for reflective links  ${{\bf{h}}_{Ii}} \sim {\cal C}{\cal N}\left( {{\bf{0}},\tau^2{{\bf{I}}_N}} \right)$ with path loss factors $\tau^2$.  
For direct links between same-side users (i.e., $i,j \in \left[ {1,K} \right]$ or $i,j \in \left[ {K + 1,2K} \right]$), ${h_{ij}} \sim {\cal C}{\cal N}\left( {0,\sigma _2^2} \right)$, while for cross-side users
($i \in \left[ {1,K} \right]$, $j \in \left[ {K + 1,2K} \right]$), ${h_{ij}} \sim {\cal C}{\cal N}\left( {0,\sigma _3^2} \right)$. 
To reflect the fact that same-side links are typically stronger due to shorter distances, we assume ${\sigma_2^2} > {\sigma_3^2}$.
The scenarios involving different channel models and user-specific path losses are simulated in Section \ref{Simulation}, and the results validate these assumptions.

\subsection{Signal Model}
The transmit signal of user $i$ is given by
${x_i} = \sqrt {{P_i}} {s_i}$,
where ${s_i}$ denotes the data symbol and ${{P_i}}$ is the corresponding transmit power.
Each ${s_i}$ is assumed to be an independent complex Gaussian symbol with zero mean and unit variance. Accordingly, the transmitted signals satisfy
$E\left( {{x}_{i}}x_{j}^{*} \right)=\left\{ \begin{matrix}
   {{P}_{i}},i=j  \\
   0,i\ne j  \\
\end{matrix} \right.$.
Thus, the signal reflected from the passive RIS is given by
\begin{equation}
{{\bf{x}}_{r}}{\rm{ = diag}}\left( {\bf{v}} \right)\sum\limits_{i = 1}^{2K} {{{\bf{h}}_{Ii}}{x_i}}, 
\end{equation}
where $\sum\limits_{i = 1}^{2K} {{{\bf{h}}_{Ii}}{x_i}} $ denotes the total signal received by the passive RIS, and
${\mathbf{v}} = {\left[ {{e^{j{\beta _1}}}, \cdots ,{e^{j{\beta _N}}}} \right]^T} \in {\mathbb{C}^{N \times 1}},{\beta_n} \in \left[ {0,2\pi } \right),\forall n$.

Note that, due to simultaneous transmission and reception, each user also experiences self-interference from its own transmit signal. However, since each user knows its own transmitted signal and the RIS configuration, 
this self-interference can be canceled\footnote{The case where self-interference cancellation is impossible is beyond the scope of this paper and left for our future work.} \cite{atapattu2020reconfigurable}.
Therefore, the received signal at user $i \in \left[ {1,K} \right]$ can be denoted as ($i \in \left[ {K+1,2K} \right]$ is similar.)
\begin{equation}
\begin{aligned}
&{{y_{i}}} = {\bf{h}}_{iI}^H{\rm{diag}}\left( {\bf{v}} \right)\sum\limits_{j \ne i}^{2K} {{{\bf{h}}_{Ij}}{x_j}}  + \sum\limits_{j \ne i}^{2K} {{h_{ij}}{x_j}}  + {n_i}\\
&  = \left( {{\bf{a}}_{i\left( {K + i} \right)}^H{\bf{v}} + {h_{i\left( {K + i} \right)}}} \right){x_{\left( {K + i} \right)}} + {n_i}\\
& + \underbrace {\sum\limits_{j \in [1,K],\atop
j \ne i} {\left( {{\bf{a}}_{ij}^H{\bf{v}} + {h_{ij}}} \right){x_j}} }_{{\text{Same-side interference}}} + \underbrace {\sum\limits_{j \in [K + 1,2K],\atop
j \ne K + i} {\left( {{\bf{a}}_{ij}^H{\bf{v}} + {h_{ij}}} \right){x_j}} }_{{\text{Cross-side interference}}},
\end{aligned}
\label{S1}
\end{equation}
where ${n_i} \sim \mathcal{C}\mathcal{N}\left( {0,\sigma^2} \right)$ is the additive white Gaussian noise (AWGN) at the user $i$, ${\mathbf{a}}_{ij}^H = {\mathbf{h}}_{iI}^H\rm{diag}\left( {{{\mathbf{h}}_{Ij}}} \right) \in {\mathbb{C}^{1 \times N}}$ is the cascaded reflective channel from user $j$ to user $i$ and can be approximated as Gaussian distribution\footnote{Although each element of the cascaded channel is the product of two Rayleigh fading coefficients and is therefore not standard Gaussian, its mean and variance can be exactly characterized. Moreover, since the subsequent analysis depends primarily on the first- and second-order statistics of the effective channel, it is analytically tractable to adopt a moment-matched complex Gaussian approximation for each cascaded channel element. This approximation is further validated by numerical simulations.}, 
i.e., ${{\bf{a}}_{ij}} \sim {\cal C}{\cal N}\left( {{\bf{0}},\sigma _1^2{{\bf{I}}_N}} \right)$ with ${\sigma _1} = \tau \times \tau = {\tau^2}$ \cite{9854102}. 

\section{Required Number of REs for Full-DoF}
\label{PRE}
This section begins by formulating an equivalent feasibility problem for achieving interference-free transmission. Then, the necessary and sufficient conditions for its solvability, i.e., the required number of REs, are established by analyzing the geometric structure of the interference-nulling equations and the unit-modulus constraints.

\subsection{Equivalent Feasibility Problem}
To achieve interference-free transmission, all signals except those from the intended pair should be eliminated for each user. Specifically, for user $i \in \left[ {1,K} \right]$, the following conditions must be satisfied simultaneously ($i \in \left[ {K+1,2K} \right]$ is similar.)
\begin{equation}
\left\{ {\begin{aligned}
&\left( {{\mathbf{a}}_{i\left( {K + i} \right)}^H{\mathbf{v}} + {h_{i\left( {K + i} \right)}}} \right){x_{\left( {K + i} \right)}} \ne 0,\\
&\left( {{\mathbf{a}}_{ij}^H{\mathbf{v}} + {h_{ij}}} \right){x_j} = 0,\forall j \in \left[ {1,K} \right],j \ne i,\\
&\left( {{\mathbf{a}}_{ij}^H{\mathbf{v}} + {h_{ij}}} \right){x_j} = 0,\forall j \in \left[ {K + 1,2K} \right],j \ne K + i,
\end{aligned}} \right.
\label{eq2}  
\end{equation}
where the first term denotes the desired signal from the intended pair (i.e., user $K+i$), and the second and third terms represent the interference from same-side and cross-side users, respectively.
Given the reciprocity of these channels (i.e., ${h_{ij}} = {h_{ji}}$, ${{\bf{a}}_{ij}} = {{\bf{a}}_{ji}}$), each side (left or right) has $\frac{K(K-1)}{2}$ distinct same-side interference nulling equations among its users, while between the two sides, there are $K\left( {K - 1} \right)$ distinct cross-interference nulling equations.

Therefore, these conditions for achieving interference-free transmission can be transformed into a feasibility problem as
\begin{equation}
\begin{gathered}
{\text{find }}{\mathbf{v}} \hfill \\
{\text{s}}{\text{.t}}{\text{. }}{\mathbf{v}} \in {S_1} \cap {S_2} \hfill, \\ 
\end{gathered}
\label{eq3}
\end{equation}
where
\begin{equation}
{S_1} = \left\{ {{\bf{v}}\left| {\left\{ {\begin{array}{*{20}{c}}
{{\bf{A}}_1^H{\bf{v}} + {{\bf{b}}_1} = 0}\\
{{\bf{A}}_2^H{\bf{v}} + {{\bf{b}}_2} = 0}
\end{array}} \right.} \right.} \right\} \buildrel \Delta \over = \left\{ {{\bf{v}}\left| {{{\bf{A}}^H}{\bf{v}} + {\bf{b}} = 0} \right.} \right\}
\label{eq4}, 
\end{equation}

\begin{equation}
{S_2} = \left\{ {\left. {\bf{v}} \right|\left| {{v_i}} \right| = 1,\forall i = 1, \cdots ,N} \right\}
\label{eq5},  
\end{equation}
and ${\bf{A}} = [{{\bf{A}}_1},{{\bf{A}}_2}] \in {\mathbb{C}^{N \times 2K\left( {K - 1} \right)}}$, ${\bf{b}} = {[{\bf{b}}_1^T,{\bf{b}}_2^T]^T} \in {\mathbb{C}^{2K\left( {K - 1} \right) \times 1}}$,
\begin{equation}
{{\mathbf{A}}_1} = \left[ {{{\mathbf{a}}_{12}}, \cdots ,{{\mathbf{a}}_{(K - 1)K}},{{\mathbf{a}}_{K + 1,(K + 2)}} \cdots ,{{\mathbf{a}}_{\left( {2K - 1} \right),2K}}} \right] ,
\nonumber
\end{equation}
\begin{equation}
{{\mathbf{A}}_2} = \left[ {{{\mathbf{a}}_{1,(K + 2)}}, \cdots ,{{\mathbf{a}}_{1,2K}}, \cdots ,{{\mathbf{a}}_{K,K + 1}} \cdots ,{{\mathbf{a}}_{K,2K - 1}}} \right] ,
\nonumber
\end{equation}
\begin{equation}
{{\mathbf{b}}_1} = {\left[ {{h_{12}}, \cdots ,{h_{(K - 1)K}},{h_{K + 1,(K + 2)}} \cdots ,{h_{\left( {2K - 1} \right),2K}}} \right]^T} ,
\nonumber
\end{equation}
\begin{equation}
{{\mathbf{b}}_2} = {\left[ {{h_{1,(K + 2)}}, \cdots ,{h_{1,2K}}, \cdots ,{h_{K,K + 1}} \cdots ,{h_{K,2K - 1}}} \right]^T} .
\nonumber
\end{equation}
Each column of ${{\mathbf{A}}_1} \in {\mathbb{C}^{N \times K\left( {K - 1} \right)}}$ and ${{\mathbf{A}}_2} \in {\mathbb{C}^{N \times K\left( {K - 1} \right)}}$ is independently distributed as ${{\mathbf{a}}_{ij}} \sim \mathcal{C}\mathcal{N}\left( {0,\sigma _1^2{{\mathbf{I}}_N}} \right)$. And ${{\bf{b}}_1} \sim {\cal C}{\cal N}\left( {0,\sigma _2^2{{\bf{I}}_{K\left( {K - 1} \right)}}} \right)$, ${{\bf{b}}_2} \sim {\cal C}{\cal N}\left( {0,\sigma _3^2{{\bf{I}}_{K\left( {K - 1} \right)}}} \right)$. 

Due to the randomness of the cascaded links ${\bf{a}}_{ij}$ and the direct links $h_{ij}$, problem \eqref{eq3} can be solved when the dimension of ${\mathbf{v}}$ (i.e., $N$) is sufficiently large \cite{bafghi2022degrees,jiang2022interference}. To this end, we employ the alternating projection algorithm \cite{jiang2022interference} to solve it as follows
\begin{equation}
\left\{ {\begin{aligned}
	&{{\mathbf{\hat v}} = {{\mathbf{v}}^{\left( t \right)}} - {\mathbf{A}}{{\left( {{{\mathbf{A}}^H}{\mathbf{A}}} \right)}^{ - 1}}\left( {{{\mathbf{A}}^H}{{\mathbf{v}}^{\left( t \right)}} + {\mathbf{b}}} \right)} \\ 
	&{{\mathbf{v}}^{\left( {t + 1} \right)}} = {{{\mathbf{\hat v}}.} \mathord{\left/
 {\vphantom {{{\mathbf{\hat v}}.} {\left\| {{\mathbf{\hat v}}} \right\|}}} \right.
 \kern-\nulldelimiterspace} {\left\| {{\mathbf{\hat v}}} \right\|}} 
	\end{aligned}} \right.
\label{eq6},
\end{equation}
where ${{\mathbf{v}}^{\left( t \right)}}$ represents the vector ${\mathbf{v}}$ at the $t$-th iteration, and the second step is the element-wise normalization. The convergence of this algorithm has been proven in \cite{jiang2022interference} and is not repeated here.

After obtaining the optimal $\mathbf{v}$, we adopt the DoF as the performance metric for interference elimination, as it reflects the total number of interference-free channels the system can support simultaneously and serves as a first-order approximation of the system capacity \cite{6197715}. The DoF is given by
\begin{equation}
    {\rm{DoF}} = \sum\limits_{i = 1}^{2K} {\mathop {\lim }\limits_{{\rm{SN}}{{\rm{R}}_i} \to \infty } \frac{{{R_i}}}{{\log \left( {{\rm{SN}}{{\rm{R}}_i}} \right)}}}
\label{eq7},
\end{equation}
where ${{\rm{SN}}{{\rm{R}}_i}}$ is the received SNR of user $i$, and the rate ${R_i}$ of user $i$ ($i \in \left[ {1,K} \right]$) is given by 
\begin{equation}
{R_i} = \log \left( {1 + \frac{{{P_i}{{\left| {{\mathbf{a}}_{i\left( {K + i} \right)}^H{\mathbf{v}} + {h_{i\left( {K + i} \right)}}} \right|}^2}}}{{\sum\limits_{{j \in [1,2K],} \atop{j \ne i,K + i} } {{P_j}{{\left| {{\mathbf{a}}_{ij}^H{\mathbf{v}} + {h_{ij}}} \right|}^2}}  + {\sigma ^2}}}} \right)
\label{eq8}
\end{equation}     
The rate for users $i \in \left[ {K+1,2K} \right]$ is similar. Therefore, when interference is completely eliminated, a \textit{full DoF} of $2K$ can be achieved for the two-way $K$-user channel.

\subsection{Necessary Condition on REs for Full DoF}
\label{Necessary}
Considering that ${S_2}$ is a subset of the sphere and ${S_1}$ denotes the solution space for the interference nulling equations, if the radius of the sphere is less than the distance from its center to this solution space, the sphere cannot intersect ${S_1}$. In this subsection, we derive the necessary $N$ for achieving the full DoF based on this geometric property.

Specifically, each point ${\mathbf{v}}$ in the set ${S_2}$ satisfies ${{\mathbf{v}}^H}{\mathbf{v}} = N$, indicating that ${S_2}$ lies on the complex sphere of radius $r=\sqrt N$. 
The set ${S_1}$ can be equivalently expressed as
\begin{equation}
{S_1}:{{\mathbf{A}}^H}{\mathbf{v}}{\text{ + }}{\mathbf{b}} = 0 \Rightarrow {\mathbf{v}} =  - {\left( {{{\mathbf{A}}^H}} \right)^ + }{\mathbf{b}} + {{\mathbf{v}}_0}
\label{eq9},
\end{equation}
where ${\left( {{{\mathbf{A}}^H}} \right)^ + } = {\mathbf{A}}{\left( {{{\mathbf{A}}^H}{\mathbf{A}}} \right)^{ - 1}}$ and ${{\mathbf{v}}_0} \in {\text{null}}\left( {{{\mathbf{A}}^H}} \right)$. Thus, ${S_1}$ can be viewed as an affine subspace obtained by translating the null space of ${\bf{A}}^H$ by the vector ${\left( {{{\mathbf{A}}^H}} \right)^ + }{\mathbf{b}} \in {\mathbb{C}^N}$. The translation distance\footnote{The translation distance of ${\text{null}}\left( {{{\mathbf{A}}^H}} \right)$ is defined as the norm of the projection of the translation vector onto the orthogonal complement of ${\text{null}}\left( {{{\mathbf{A}}^H}} \right)$, which is spanned by the columns of $\mathbf{A}$. Thus, ${\text{Pro}}{{\text{j}}_{{{\left( {{\text{null}}\left( {{{\mathbf{A}}^H}} \right)} \right)}^ \bot }}}\left(  \cdot  \right) = {\mathbf{A}}{\left( {{{\mathbf{A}}^H}{\mathbf{A}}} \right)^{ - 1}}{{\mathbf{A}}^H}\left(  \cdot  \right)$.} $d$ is given by 
\begin{equation}
\begin{aligned}
d &= \left| {{\text{Pro}}{{\text{j}}_{{{\left( {{\text{null}}\left( {{{\mathbf{A}}^H}} \right)} \right)}^ \bot }}}\left( {{{\left( {{{\mathbf{A}}^H}} \right)}^ + }{\mathbf{b}}} \right)} \right| \hfill \\
&= \left| {{\mathbf{A}}{{\left( {{{\mathbf{A}}^H}{\mathbf{A}}} \right)}^{ - 1}}{{\mathbf{A}}^H} \cdot \left( {{\mathbf{A}}{\left( {{{\mathbf{A}}^H}{\mathbf{A}}} \right)^{ - 1}}{\mathbf{b}}} \right)} \right| \hfill  
= \left| {{{\left( {{{\mathbf{A}}^H}} \right)}^ + }{\mathbf{b}}} \right| \hfill . 
\end{aligned}   
\end{equation}

To determine a necessary condition on $N$ for the solvability of problem \eqref{eq3}, we establish the following theorem.
\begin{theorem}
Let
${{\bar S}_1} = \left\{ {\left. {{\mathbf{x}} \in {\mathbb{C}^{N}}} \right|{\mathbf{x}} = {{\mathbf{x}}_0} + {\mathbf{z}},{\mathbf{z}} \in {\text{null}}\left( {\mathbf{G}} \right)} \right\}$, \\
${{\bar S}_2} = \left\{ {\left. {{\mathbf{x}} \in {\mathbb{C}^{N}}} \right|\left| {{x_i}} \right| = \gamma ,\forall i = 1, \cdots ,N} \right\}$, 
where ${\mathbf{G}} \in {\mathbb{C}^{L \times N}}$ and ${{\mathbf{x}}_0} \in {\mathbb{C}^{N}}$ are given, and $L \leqslant N$.
If $\left| {{\text{Pro}}{{\text{j}}_{{{\left( {{\text{null}}\left( {\mathbf{G}} \right)} \right)}^ \bot }}}\left( {{{\mathbf{x}}_0}} \right)} \right|  >  \gamma\sqrt N $, then ${{\bar S}_1} \cap {{\bar S}_2} = \varnothing  $
\label{the1}
\end{theorem}
\begin{proof}
	Please see Appendix~\ref{appA}.
\end{proof}

One geometric interpretation of Theorem~\ref{the1} is that if the distance from the affine subspace to the origin exceeds the radius of the sphere, then it cannot intersect the sphere, and consequently cannot intersect any subset lying on the sphere.

Therefore, by applying Theorem~\ref{the1}, we have that if 
\begin{equation}
\sqrt N  < \left| {{{\left( {{{\mathbf{A}}^H}} \right)}^ + }{\mathbf{b}}} \right|
\label{eq10}, 
\end{equation}
then ${S_1} \cap {S_2} = \varnothing$. 
Given the randomness of ${\mathbf{A}}$ and ${\mathbf{b}}$, ${{{\left( {{{\bf{A}}^H}} \right)}^ + }{\bf{b}}}$ 
is also a random vector. According to the norm concentration in high-dimensional statistics \cite{vershynin2018high}, when the dimension of vector ${{{\left( {{{\bf{A}}^H}} \right)}^ + }{\bf{b}}}$ is large, the distribution of its norm will be highly concentrated around its expected value. And the expected value of $\left| {{{\left( {{{\bf{A}}^H}} \right)}^ + }{\bf{b}}} \right|$ can be calculated as
\begin{equation}
\begin{aligned}
E\left( {\left| {{{\left( {{{\bf{A}}^H}} \right)}^ + }{\bf{b}}} \right|} \right) & \mathop  \approx \limits^{\left( a \right)} \sqrt {E\left( {{{\left( {{{\left( {{{\mathbf{A}}^H}} \right)}^ + }{\mathbf{b}}} \right)}^H}{{\left( {{{\mathbf{A}}^H}} \right)}^ + }{\mathbf{b}}} \right)}  \\
 &= \sqrt {E\left( {\text{tr}\left( {{\bf{b}}{{\bf{b}}^H}{{\left( {{{\bf{A}}^H}{\bf{A}}} \right)}^{ - 1}}} \right)} \right)} \\
 &\mathop  = \limits^{(b)}  \sqrt {\frac{{K\left( {K - 1} \right)\left( {\eta _1^2 + \eta _2^2} \right)}}{{N - 2K\left( {K - 1} \right) - 1}}}, 
\end{aligned}
\label{eq11}    
\end{equation}
where ${\eta _1} = \frac{{{\sigma _2}}}{{{\sigma _1}}}$ and ${\eta _2} = \frac{{{\sigma _3}}}{{{\sigma _1}}}$ are the strength ratios of the same-side and cross-side direct links to the cascaded reflective link, respectively. The approximation (a) follows from the concentration of $\left| {{{\left( {{{\bf{A}}^H}} \right)}^ + }{\bf{b}}} \right|$ in the large-dimensional regime
\cite{vershynin2018high}. The equality (b) follows from the fact that ${{{\left( {{{\mathbf{A}}^H}{\mathbf{A}}} \right)}^{ - 1}}}$ follows an inverse Wishart distribution.

Consequently, the inequality \eqref{eq10} can be transformed as
\begin{equation}
\sqrt N  < \sqrt {\frac{{K\left( {K - 1} \right)\left( {\eta _1^2 + \eta _2^2} \right)}}{{N - 2K\left( {K - 1} \right) - 1}}}
\label{eq12}.   
\end{equation}
It can be reformulated as
\begin{equation}
\begin{aligned}
N < &\frac{1}{2}\sqrt {{{\left( {2K\left( {K - 1} \right) + 1} \right)}^2} + 4K\left( {K - 1} \right)\left( {\eta _1^2 + \eta _2^2} \right)} \\
 &+ \frac{1}{2}\left( {2K\left( {K - 1} \right) + 1} \right)\triangleq {N_\text{nec}}
\end{aligned}
\label{eq13}.   
\end{equation}
Furthermore, this necessary threshold ${N_\text{nec}}$ can be asymptotically characterized as
\begin{equation}
    {N_\text{nec}} = \left\{ {\begin{aligned}
&O\left( {{K^2}} \right),\;{\text{if}}\:\eta < K\\
&O\left( {K\eta } \right),\;{\text{if}}\:\eta > {K}
\end{aligned}} \right.
\label{eq14},
\end{equation}
where $\eta  = \sqrt {\eta _1^2 + \eta _2^2}$, $\eta_1$ and $\eta_2$ can be
taken as the strength ratios of the weakest same-side and cross-side direct links to the strongest cascaded reflective link, respectively.
If the number of passive REs is below this threshold, problem \eqref{eq3} is unlikely to have a solution (i.e., ${S_1} \cap {S_2} = \varnothing$). Therefore, this necessary condition can also serve as a lower bound on the required number of passive REs to achieve the full DoF.

\subsection{Sufficient Condition on RE for Full DoF}
\label{Sufficient}
We have shown that if the sphere radius $r$ is smaller than the distance $d$ from its center to $S_1$, then the sphere does not intersect ${S_1}$, resulting in ${S_1} \cap {S_2} = \varnothing$. However, when $r$ slightly exceeds $d$, the sphere can intersect ${S_1}$, but ${S_2}$ may still not, as it is merely a subset of the sphere.
By how much must $r$ exceed $d$ to ensure that ${S_2}$ intersects ${S_1}$? This subsection aims to address this question and derive a sufficient condition.

To gain insight into the relationship between ${S_1}$ and ${S_2}$, a graphical illustration is provided in Fig. \ref{d}, where the red dots represent the set ${S_2}$ and the blue plane represents the set ${S_1}$, all of them lie in the high-dimensional complex space.
As shown in Fig. \ref{d}, when the complex sphere intersects the solution space ${S_1}$, their cross-section forms a complex circle (or more generally, a hyperspherical section, depending on the dimensionality of ${S_1}$ and the sphere). 
This cross-section (denoted by ${A'}{O'}{B'}$) divides the sphere into two regions: the yellow spherical cap above and the remaining part. The two endpoints ${A'}$ and ${B'}$ of the diameter of the cross-section form an angle $\theta $ with the origin. Thus, the relationship between $r$ and $d$ can be characterized by the central angle $\theta $ as follows
\begin{equation}
\cos \left( {\frac{\theta }{2}} \right) = \frac{d}{r}.
\label{eq15}   
\end{equation}

\begin{figure}
\centering
\subfigure[$\theta  \ll \frac{\pi }{2}$]{
\includegraphics[width=0.35\linewidth]{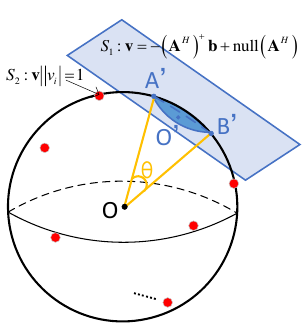}
\label{d1}
}
\quad
\subfigure[$\theta  = \frac{\pi }{2}$]{
\includegraphics[width=0.35\linewidth]{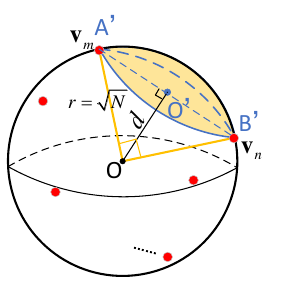}
\label{d2}
}

\caption{Intersection between $S_1$ and the complex sphere ($r>d$). }
\label{d}
\end{figure}

Furthermore, if we take any two points from the set ${S_2}$, such as ${{\bf{v}}_m}$ and ${{\bf{v}}_n}$, then we have
\begin{equation}
\begin{aligned}
E\left( {{\mathbf{v}}_n^H{{\mathbf{v}}_m}} \right) &= E\left( {\sum\limits_{i = 1}^N {v_{n,i}^*{v_{m,i}}} } \right) = NE\left( {v_{n,i}^*{v_{m,i}}} \right) \hfill \\
  &\mathop  = \limits^{(a)} NE\left( {v_{n,i}^*} \right)E\left( {{v_{m,i}}} \right) \hfill \\
  &\mathop  = \limits^{(b)} N\int_0^{2\pi } {\frac{1}{{2\pi }}{e^{j{\varphi _n}}}d{\varphi _n}}  \cdot \int_0^{2\pi } {\frac{1}{{2\pi }}{e^{j{\varphi _m}}}d{\varphi _m}}= 0 \hfill \\ 
\end{aligned}
\label{eq16},  
\end{equation}
where $\left( a \right)$ follows from the fact that the entries of $\mathbf{v} \in S_2$ are independent, and $\left( b \right)$ follows from the fact that the phase $\varphi_i$ of each element ${v_i}$ in ${\mathbf{v}} \in {S_2}$ takes values uniformly\footnote{Although RIS phase shifts should be carefully designed for interference nulling, for feasibility analysis we treat each element of ${\mathbf{v}}\in S_2$ as a generic phase uniformly distributed over $\left[0,2\pi\right)$ to characterize the geometry of the unit-modulus set $S_2$. Interference-free transmission exists if and only if ${S_1}\cap S_2 \ne \varnothing$, which depends only on the relative geometry of the two sets rather than any specific phase-shift design.}
over $\left[ {0,2\pi } \right)$ \cite{zhu2024robust}.
According to the law of large numbers, equation \eqref{eq16} implies that as $N \to \infty $,  
any two distinct vectors in the set ${S_2}$ are orthogonal.
Therefore,  when the endpoints (${A'}$ and ${B'}$) of the cross-section's diameter coincides with any two points ${{\bf{v}}_m}$ and ${{\bf{v}}_n}$, as shown in Fig. \ref{d2}, we have $\theta  = \frac{\pi }{2}$. In this case, the relationship between $r$ and $d$ can be characterized as 
\begin{equation}
d = \cos \left( {\frac{\pi }{4}} \right)r = \frac{r}{{\sqrt 2 }}
\label{eq17}. 
\end{equation}

To determine a sufficient condition on $N$ to achieve the full DoF, we establish the following theorem.
\begin{theorem}
If $d \leqslant \frac{r}{{\sqrt 2 }}$ or $\theta    \geqslant   \frac{\pi }{2}$, then the cross-section of the sphere and $S_1$ must pass through at least one point of the set ${S_2}$, i.e., ${S_1} \cap {S_2} \ne \varnothing $.
\label{the2}
\end{theorem}
\begin{proof}
Please see Appendix \ref{appB}.
\end{proof}
Therefore, by Theorem~\ref{the2}, the sufficient condition on $N$ for ${{S}_1} \cap {{S}_2} \ne \varnothing $ can be derived by
\begin{equation}
\begin{aligned}
r \geqslant \sqrt 2 d \Leftrightarrow \sqrt N  \geqslant \sqrt {\frac{{2K\left( {K - 1} \right)\left( {\eta _1^2 + \eta _2^2} \right)}}{{N - 2K\left( {K - 1} \right) - 1}}} .
\end{aligned}
\label{eq21} 
\end{equation}
It can be reformulated as
\begin{equation}
\begin{aligned}
N \geqslant &\frac{1}{2}\sqrt {{{\left( {2K\left( {K - 1} \right) + 1} \right)}^2} + 8K\left( {K - 1} \right)\left( {\eta _1^2 + \eta _2^2} \right)} \\
 &+ \frac{1}{2}\left( {2K\left( {K - 1} \right) + 1} \right) \triangleq {N_\text{suf}}
\end{aligned}
  \label{eq22}.  
\end{equation}
This sufficient threshold ${N_\text{suf}}$ for the required number of passive REs can also be asymptotically characterized as
\begin{equation}
    {N_\text{suf}} = \left\{ {\begin{aligned}
&O\left( {{K^2}} \right),\;{\text{if}}\:\eta < K\\
&O\left( {K\eta } \right),\;{\text{if}}\:\eta > {K}
\end{aligned}} \right.
,
\end{equation}
where $\eta  = \sqrt {\eta _1^2 + \eta _2^2}$, $\eta_1$ and $\eta_2$ 
can be taken as the ratios of the strongest same-side and cross-side direct links to the weakest cascaded reflective link, respectively.
If the number of passive REs $N$ exceeds this threshold, the feasibility of problem \eqref{eq3} is ensured. Therefore, this sufficient condition can serve as an upper bound on the REs to achieve the full DoF.

\section{Degrees of Freedom under Imperfect CSI}
\label{Law}
In the previous sections, we investigated the required number of passive REs to achieve interference-free transmission (i.e., full DoF) under perfect CSI. In practice, however, acquiring accurate RIS-related CSI is challenging due to the passive nature of the REs, which leads to residual interference and non-negligible performance degradation. 
To characterize this degradation, we formulate a sum-rate maximization problem and derive the corresponding DoF via the Squeeze Theorem by analyzing the upper and lower bounds.

The preceding analysis assumed Rayleigh fading for tractability, as the phase transition in the required number of REs for interference nulling is independent of the channel model \cite{jiang2022interference}. In this section, to better capture the impact of the imperfect CSI in practical scenarios, we model the RIS-related channels (${\mathbf{h}}_{Ii}$) as Rician fading to account for line-of-sight components, while the direct user-to-user links (${h_{ij}}$) remain Rayleigh fading due to extensive scattering \cite{pan2020multicell}.
Specifically, 
\begin{equation}
{h_{ij}} \sim \mathcal{C}\mathcal{N}\left( {0,{L_{ij}}} \right)
\label{eq_hij}, 
\end{equation}
\begin{equation}
{{\mathbf{h}}_{Ii}} = \sqrt {{L_{Ii}}} \left( {\sqrt {\frac{\varepsilon }{{\varepsilon  + 1}}} {{{\mathbf{\bar h}}}_{Ii}}^{{\text{LoS}}} + \sqrt {\frac{1}{{\varepsilon  + 1}}} {{{\mathbf{\tilde h}}}_{Ii}}^{{\text{NLoS}}}} \right)
\label{eq_hIi},
\end{equation}
where $L$ is the distance-dependent path loss, $\varepsilon$ is the Rician factor, and ${{{{\mathbf{\bar h}}}_{Ii}}^{{\text{LoS}}}}$ and ${{{{\mathbf{\tilde h}}}_{Ii}}^{{\text{NLoS}}}}$ represent the line-of-sight (LoS) and non-LoS (NLoS) components, respectively, as in \cite{jiang2021learning}. 

\subsection{Sum Rate Maximization}
In practice, the direct channels (${h_{ij}}$) between users can be estimated by turning off the RIS. Thus, it is reasonable to assume perfect CSI for direct channels \cite{ZhouGui,zhou2020framework}. However, obtaining accurate CSI for the passive RIS-related channels (${{\mathbf{h}}_{Ii}}$) is challenging. Note that the cascaded channels (${{\mathbf{a}}_{ij}}$), defined as ${\mathbf{a}}_{ij}^H = {\mathbf{h}}_{Ii}^H\rm{diag}\left( {{{\mathbf{h}}_{Ij}}} \right)$, are sufficient for RIS phase shift design \cite{li2023performance,zhou2020framework}. We aggregate the RIS-related CSI errors into the cascaded channels \cite{gao2023robust,10795157}, modeled as
\begin{equation}
{{\mathbf{a}}_{ij}} = \varsigma {{{\mathbf{\hat a}}}_{ij}} + \sqrt {1 - {\varsigma ^2}} \Delta {{\mathbf{a}}_{ij}}
\label{eq_er},
\end{equation}
where ${{\mathbf{a}}_{ij}}$, ${{{\mathbf{\hat a}}}_{ij}}$, and $\Delta {{\mathbf{a}}_{ij}}$ denote the ideal channel, the estimated channel, and the channel estimation error, respectively. The accuracy factor $0 \leq \varsigma \leq 1$ characterizes the estimation quality, with $\varsigma=1$ denoting perfect estimation and $\varsigma=0$ denoting completely erroneous estimation \cite{zhu2024robust,9957104}.
The error term follows $\Delta {{\mathbf{a}}_{ij}} \sim \mathcal{C}\mathcal{N}\left( {0,\left( {{L_{Ii}}{L_{Ij}}} \right){{\mathbf{I}}_N}} \right)$, independent of $\hat{\mathbf{a}}_{ij}$, ensuring that the average channel power of ${\mathbf{a}}_{ij}$ remains ${{L_{Ii}}{L_{Ij}}}$ \cite{9957104,yang2021performance,vu2022performance}.

Therefore, by substituting \eqref{eq_er} into \eqref{S1}, the received signal of user $i \in \left[ {1,K} \right]$ under imperfect CSI can be expressed as
\begin{equation}
\begin{aligned}
&{{\tilde y}_i} = \underbrace {\left( {\varsigma {\mathbf{\hat a}}_{i\left( {K + i} \right)}^H{\mathbf{v}} + {h_{i\left( {K + i} \right)}}} \right){x_{\left( {K + i} \right)}} }_{{\text{Signal with estimated CSI}}} \\
&    + \underbrace {\sum\limits_{j \in [1,2K],\atop
j \ne i,K+i} {\left( {\varsigma {\bf{\hat a}}_{ij}^H{\bf{v}} + {h_{ij}}} \right){x_j}}}_{{\text{Interference with estimated CSI}}} +\underbrace {\sum\limits_{j \in [1,2K], \atop j \ne i,K+i } {\sqrt {1 - {\varsigma ^2}} \Delta {\mathbf{a}}_{ij}^H{\mathbf{v}}{x_j}}}_{{\text{Interference with CSI error ${\left( {{n_\text{eq1}}} \right)}$}}}\\
&+\underbrace {\sqrt {1 - {\varsigma ^2}} \Delta {\mathbf{a}}_{i\left( {K + i} \right)}^H{\mathbf{v}}{x_{K + i}}}_{{\text{Signal with CSI error ${\left( {{n_\text{eq2}}} \right)}$}}}  + {n_i}.
\end{aligned}
\end{equation}
The terms associated with CSI errors can be treated as equivalent noise, since only their statistical characteristics are available \cite{li2023performance}. Based on Appendix~\ref{APP_p_eq}, the power of these equivalent noise is given by
\begin{equation}
{P_{{\text{eq}}}} = E\left( {{{\left| {{n_{{\text{eq1}}}} + {n_{{\text{eq2}}}}} \right|}^2}} \right) = \left( {1 - {\varsigma ^2}} \right)\sum\limits_{j \in [1,2K], \atop j \ne i} {{P_j}N {{L_{Ii}}{L_{Ij}}} }
\label{p_eq}.
\end{equation}

Therefore, the rate of user $i \in \left[ {1,K} \right]$ under imperfect CSI can be expressed as
\begin{equation}
{{\tilde R}_i} = {\log _2}\left( {1 + \frac{{{P_{K + i}}{{\left| {\varsigma {\mathbf{\hat a}}_{i\left( {K + i} \right)}^H{\mathbf{v}} + {h_{i\left( {K + i} \right)}}} \right|}^2}}}{{\sum\limits_{j \in [1,2K],\atop j \ne i,K+i} {{{\left| {\varsigma {\mathbf{\hat a}}_{ij}^H{\mathbf{v}} + {h_{ij}}} \right|}^2}{P_j}}+ {P_{{\text{eq}}}} + {\sigma ^2}}}} \right)
\label{rate}.
\end{equation}
The rate of user $i \in \left[ {K+1,2K} \right]$ follows a similar expression and is omitted for brevity.
Accordingly, the maximum achievable capacity of the system with imperfect CSI can be obtained by solving the following optimization problem:
\begin{equation}
\begin{aligned}
  \max{\text{ }} &R_{\text{sum}} =\sum\limits_{i = 1}^{2K}{{\tilde R}_i}  \hfill \\
  {\text{s}}{\text{.t}}{\text{. }}&\left| {{v_n}} \right| = 1 \hfill \\ 
\end{aligned}
\label{max}.
\end{equation}
It can be solved by manifold optimization using the Manopt toolbox \cite{boumal2023introduction}.

\subsection{Upper Bound on DoF}
The exact expression of the maximum sum rate is intractable, since the severe inter-user interference and the unit-modulus constraints of passive REs preclude a closed-form solution for the optimal RIS phase shifts. To address this challenge, we first leverage the results in subsection~\ref{Sufficient} to ensure complete interference elimination with the estimated CSI. Then, by exploiting the modulus-1 property of the passive REs together with statistical quantities, we establish an upper bound for the maximum ergodic sum rate, from which the DoF upper bound is subsequently derived based on its definition.

Specifically, since each user is equipped with only a single antenna, the interference aligned with the estimated CSI can only be mitigated through the RIS. Based on the previous analysis, perfect cancellation of this interference can be guaranteed if the number of passive REs satisfies
\begin{equation}
\begin{aligned}
N \geqslant &\frac{1}{2}\sqrt {{{\left( {2K\left( {K - 1} \right) + 1} \right)}^2} + 8K\left( {K - 1} \right)\left( {\bar\eta _1^2 + \bar\eta _2^2} \right)} \\
 &+ \frac{1}{2}\left( {2K\left( {K - 1} \right) + 1} \right) \triangleq {N_{{\text{null}}}}
\end{aligned}
\label{num_re}.
\end{equation}
${N_{{\text{null}}}}$ can be further characterized asymptotically as
\begin{equation}
    {N_\text{null}} = \left\{ {\begin{aligned}
&O\left( {{K^2}} \right),\;{\text{if}}\:\bar\eta < K\\
&O\left( {K\bar\eta } \right),\;{\text{if}}\:\bar\eta > {K}
\end{aligned}} \right.
\end{equation}
where $\bar\eta  = \sqrt {\bar\eta _1^2 + \bar\eta _2^2}$, ${\bar\eta _1} = \frac{{{\bar\sigma _2}}}{\varsigma{{\underline{\sigma}_1}}}$ and ${\bar\eta _2} = \frac{{{\bar\sigma _3}}}{{\varsigma{\underline{\sigma}_1}}}$.  
${{\bar \sigma }_2}$ and ${{\bar \sigma }_3}$ represent the path loss of the strongest same-side and cross-side direct links, respectively.  $\underline{\sigma}_1$ represents the weakest cascaded reflective link, and $\varsigma$ is the CSI accuracy factor.

Without loss of generality, we assume that each user has the same transmit power, i.e., ${P_i} = P,\forall i \in \left[ {1,2K} \right]$.  
If the number of REs satisfies \eqref{num_re}, the optimal ${\mathbf{v}}$ can cancel all interference while providing beamforming gain to maximize the sum rate \cite{jiang2022interference}. 
Consequently, the maximum sum rate is given by
\begin{equation}
\begin{aligned}    
{R} = \sum\limits_{i = 1}^{2K} {{{\log }_2}\left( {1 + \frac{{P{{\left| {\varsigma {\mathbf{\hat a}}_{i\left( {\left| {K \pm i} \right|} \right)}^H{{\mathbf{v}}_{{\text{opt}}}} + {h_{i\left( {\left| {K \pm i} \right|} \right)}}} \right|}^2}}}{{{P_{{\text{eq}}}} + {\sigma ^2}}}} \right)} 
\end{aligned}  
\label{sumrate},
\end{equation}
where ${{{\mathbf{v}}_{{\text{opt}}}}}$ denotes the optimal phase shift obtained by \eqref{max},
$\left| {K \pm i} \right|$ denotes the index of the paired user. Specifically, if $i \in \left[ {1,K} \right]$, then its desired pair is user $K+i$; otherwise, if $i \in \left[ {K+1,2K} \right]$, its desired pair is user $i-K$.

To analytically characterize the theoretical performance of the achievable rate, we adopt the ergodic rate \cite{zhi2022ris}
and approximate the strengths of all estimated cascaded reflective links by that of an extreme one\footnote{As shown later in Appendix \ref{sumrate_up} of this subsection, the upper bound of the ergodic ${R}$ increases monotonically with the channel gains of the cascaded reflective and cross-side direct links. Therefore, approximating these links by their extreme (strongest) ones is reasonable for deriving the upper bound.}, i.e., 
\begin{equation}
{\mathbf{\hat a}}_{i\left( {\left| {K \pm i} \right|} \right)}^H = {\mathbf{\hat h}}_{iI}^H{\text{diag}}\left( {{{{\mathbf{\hat h}}}_{I\left( {\left| {K \pm i} \right|} \right)}}} \right) \triangleq {{{\mathbf{\tilde a}}}_i},
\end{equation}
with
\begin{equation}
E\left( {{{\left| {{{\tilde a}_{in}}} \right|}^2}} \right) = {L_{Id}}{L_{Iu}} \triangleq L_a,
\end{equation}
where $L_{Iu}$ and $L_{Id}$ represent the strongest (or weakest) channel gains among all uplink (${{{{\mathbf{\hat h}}}_{I\left( {\left| {K \pm i} \right|} \right)}}}$) and downlink (${{{\mathbf{\hat h}}}_{iI}}$) reflective channels, respectively.
Accordingly, the equivalent noise power $P_{\text{eq}}$ given in \eqref{p_eq} can be rewritten as
\begin{equation}
{{\tilde P}_{{\text{eq}}}} = \left( {1 - {\varsigma ^2}} \right)\left( {2K - 1} \right)PN{L_a}
\label{eq_power}.
\end{equation}
Similarly, all cross-side direct links are approximated by an extreme one, i.e.,
\begin{equation}
{h_{i\left( {\left| {K \pm i} \right|} \right)}} \triangleq {{\tilde h}_i} \sim \mathcal{C}\mathcal{N}\left( {0,\sigma _h^2} \right),\forall i \in \left[ {1,2K} \right],
\end{equation}
where ${\sigma _h^2}$ represents the channel gains of the strongest (or weakest) cross-side direct link.
Using these approximations, the upper and lower bounds of ${R}$ can be derived. 

Specifically, based on the derivations in Appendix \ref{sumrate_up} and incorporating \eqref{num_re}, we obtain that when $N \geqslant {N_{{\text{null}}}}$, the upper bound of the ergodic sum-rate ${R}$ is given by
\begin{equation}
E\left( {{{\bar R}}} \right) = 2K{\log _2}\left( {1 + \frac{{P\left( {c{\varsigma ^2}{N^2}{{\bar L}_a} + \bar \sigma _h^2} \right)}}{{\left( {1 - {\varsigma ^2}} \right)\left( {2K - 1} \right)PN{{\bar L}_a} + {\sigma ^2}}}} \right)
\label{eq_E_rate},
\end{equation}
where ${{{\bar L}_a}}$ and ${\bar \sigma _h^2}$ denote the channel gains of the strongest cascaded reflective link and the strongest cross-side direct link, respectively. $c$ is a constant defined as
\begin{equation}
c = \left\{ \begin{gathered}
  \frac{\varepsilon }{{1 + \varepsilon }}{\text{, if Rician factor }}\varepsilon  \gg 1 \hfill \\
  {\left( {1 + \frac{{{\varepsilon ^2}}}{4}} \right)^2}\frac{{{\pi ^2}}}{{16}}{\text{, if }}\varepsilon  \ll 1 \hfill \\ 
\end{gathered}  \right.
\end{equation}

Furthermore, since the CSI error scales with ${\text{SN}}{{\text{R}}^{ - \alpha }}$ \cite{khalilsarai2023fdd,joudeh2016sum}, we model the CSI accuracy factor $\varsigma$ as
\begin{equation}
1 - {\varsigma ^2} = {\text{SN}}{{\text{R}}^{ - \alpha }} \Leftrightarrow {\varsigma ^2} = 1 - {\text{SN}}{{\text{R}}^{ - \alpha }}
\label{eq_s_SNR},
\end{equation}
where the CSI error decay exponent $\alpha  \in \left[ {0,\infty } \right)$. 

Therefore, substituting \eqref{eq_s_SNR} into \eqref{eq_E_rate} and applying the definition of DoF, we obtain
\begin{equation}
{\text{Do}}{{\text{F}}_{{\text{up}}}} = \mathop {\lim }\limits_{{\text{SNR}} \to \infty } \frac{{E\left( {\bar R} \right)}}{{{{\log }_2}\left( {{\text{SNR}}} \right)}} = 2K\min \left( {\alpha ,1} \right)
\label{eq_DoF_up},
\end{equation}
the detailed derivation is provided in Appendix \ref{APP_DoF_up}.

\subsection{Lower Bound on DoF}
However, the achievability of the DoF upper bound is difficult to prove in a single-antenna multi-user system. This is because the optimal phase-shift design must simultaneously cancel interference and align signal channels to maximize the sum-rate, while the modulus-1 constraints of passive REs make a closed-form solution intractable. To facilitate analysis under imperfect CSI, we now focus on the DoF lower bound and apply the Squeeze Theorem to determine the system DoF.

Specifically, similar to \eqref{eq_E_rate}, when the number of passive REs satisfies $N \geqslant {N_{{\text{null}}}}$ as given in \eqref{num_re}, the lower bound of the ergodic maximum sum-rate ${R}$ is given by
\begin{equation}
E\left( {\underline{R}} \right) = \frac{{ - 2K}}{{\ln 2}}{e^{\frac{1}{{\hat P\lambda }}}}{\text{Ei}}\left( { - \frac{1}{{\hat P\lambda }}} \right),     
\end{equation}
where $\hat P\lambda  = \frac{{P\left( {{\varsigma ^2}N{L_a} + \sigma _h^2} \right)}}{{{{\tilde P}_{{\text{eq}}}} + {\sigma ^2}}}$, and ${\text{Ei}}\left(  \cdot  \right)$ denotes the exponential integral.
The detailed derivation is provided in Appendix \ref{APP_rate_low}.

Subsequently, based on the derivations in Appendix \ref{DoF_low}, the DoF lower bound is given by
\begin{equation}
{\text{DoF}}_{{\text{low}}} = \mathop {\lim }\limits_{{\text{SNR}} \to \infty } \frac{{E\left( {\underline{R} } \right)}}{{{{\log }_2}\left( {{\text{SNR}}} \right)}} = 2K\min \left( {\alpha ,1} \right)
\label{eq_DoF_low}.
\end{equation}
By combining it with ${\text{Do}}{{\text{F}}_{{\text{up}}}}$ in \eqref{eq_DoF_up}, the system DoF under imperfect CSI is bounded as ${\text{Do}}{{\text{F}}_{{\text{low}}}} \leqslant {\text{DoF}} \leqslant {\text{Do}}{{\text{F}}_{{\text{up}}}}$, i.e.,
\begin{equation}
\begin{aligned}
2K\min \left( {\alpha ,1} \right) \leqslant {\text{DoF}} \leqslant 2K\min \left( {\alpha ,1} \right)\\ \Leftrightarrow {\text{DoF}} = 2K\min \left( {\alpha ,1} \right).
\end{aligned}
\end{equation}

\section{Simulation Result}
\label{Simulation}
In this section, we present numerical results to validate our theoretical findings. Since Theorems~\ref{the1} and~\ref{the2} constitute the theoretical basis of our work, we first consider an idealized scenario with identical path losses within each channel type but different across types, to assess the accuracy of their predictions. We then examine a more practical Rician channel model, in which users experience varying path losses, to demonstrate the applicability of our results under realistic conditions.

\subsection{Simulation Setup}
Without loss of generality, we assume that all channels are reciprocal. 
Unless specified otherwise, the system bandwidth is $1{\rm{MHz}}$, and the noise spectral density is $ - {\rm{174dBm/Hz}}$. All users transmit with identical power $P = 30\text{dBm}$. The left and right side users are evenly distributed along the $x$-axis within $\left[ {5\rm{m},50\rm{m}} \right]$ at $y=-15\rm{m}$ and $y=15\rm{m}$, respectively, both sharing a $z$-coordinate of $-20\rm{m}$. The RIS is located at $\left( {{\rm{0,0,0}}} \right)$. And 200 independent trials are conducted.

In Subsection~\ref{A}, to verify the accuracy of Theorems~\ref{the1} and~\ref{the2}, which play key roles in our theoretical derivations, we model the reflection channels as ${{\bf{h}}_{Ii}} \sim {\cal C}{\cal N}\left( {0,\tau^2{{\bf{I}}_N}} \right)$, where the path losses are set to $\tau^2 =- 30{\rm{dBm}}$.
The direct channels are modeled as ${h_{ij}} \sim \mathcal{CN}(0, (\eta_1 \sigma_1)^2)$ for same-side users and ${h_{ij}} \sim \mathcal{CN}(0, (\eta_2 \sigma_1)^2)$ for cross-side users, where $\sigma_1 = \tau^2$ and $\eta_1$, $\eta_2$ are the respective strength ratios to the cascaded reflective links.

In the remaining subsections, to evaluate the applicability of our theoretical results in practical scenarios with heterogeneous path losses caused by user locations, the direct links and RIS-related links are modeled as in \eqref{eq_hij} and \eqref{eq_hIi}, respectively. 
The Rician factor $\varepsilon$ is set to 10. The large-scale path loss $L$ in dB is modeled as \cite{jiang2021learning,ZhouGui}
\begin{equation}
L\left( \text{d} \right) =  - 30 - 10\omega {\log _{10}}\left( \text{d} \right), 
\end{equation}
where $\text{d}$ is the link distance in meters.
The path-loss exponent $\omega$ for direct links is set to $4$ to capture the higher attenuation in urban or obstructed environments. For reflective links, $\omega$ is set to $2$, reflecting the fact that RIS is typically deployed to create favorable propagation conditions and enhance transmission quality \cite{oa2018determination}.

\subsection{Accuracy Verification of Theorems~\ref{the1} and~\ref{the2}}
\label{A}
In Fig.~\ref{eta2_30}, we evaluate the theoretical thresholds $N_{\text{nec}}$ and $N_{\text{suf}}$ under an idealized scenario with uniform path loss assumed for each link category.
These thresholds, derived directly from Theorems~\ref{the1} and~\ref{the2} (as in \eqref{eq13} and \eqref{eq22}), are denoted as “necessary” and “sufficient,” respectively. The curves labeled “$99\%$ probability of full-DoF” and “$1\%$ probability of full-DoF” indicate the values of $N$ at which full DoF is achieved with $99\%$ and $1\%$ probability.


It can be observed that the simulation results closely match the theoretical bounds. In particular, as the link-strength ratio $\eta$ increases, the required number of REs grows consistently with our theoretical predictions, as stronger direct links induce more severe interference that must be eliminated by additional REs. Moreover, the required number of passive REs falls within the predicted range. However, for small $\eta$, which leads to a small $N$, some deviations from the theoretical sufficient value occur. This is because the sufficient condition is derived under the assumption of $N \to \infty$, which may lead to reduced accuracy when $N$ is small.

\begin{figure}
	\centering\includegraphics[width=6cm]{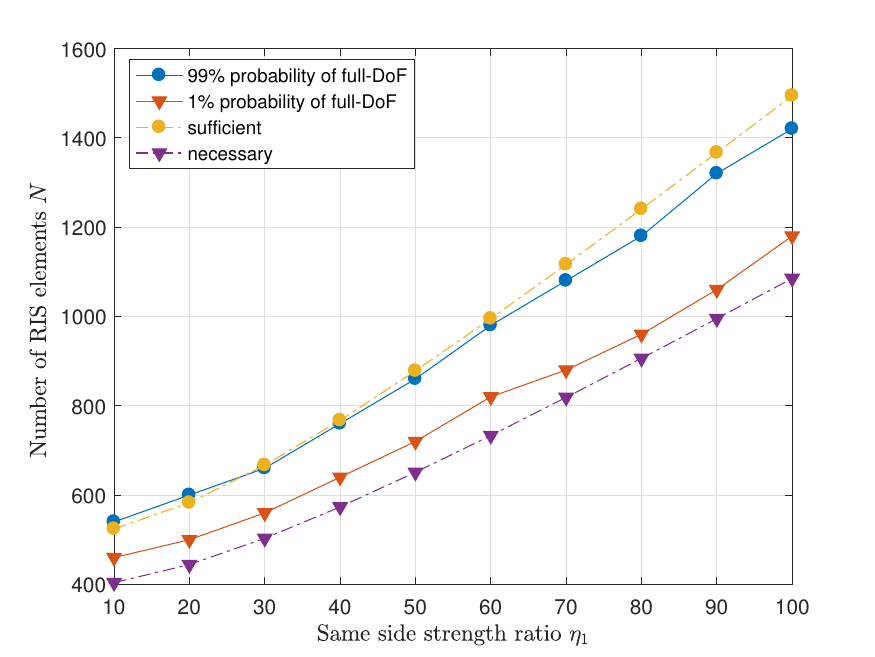}
    \caption{Sufficient and necessary conditions on REs under idealized channel settings ($K=10$, $\eta_2=30$ ).}
	\label{eta2_30}
\end{figure}

\subsection{Required Number of REs in Practical Channels}
\label{different_path}

\begin{figure*}
\centering
\subfigure[Total DoF vs. number of REs.]{
\includegraphics[width=6cm]{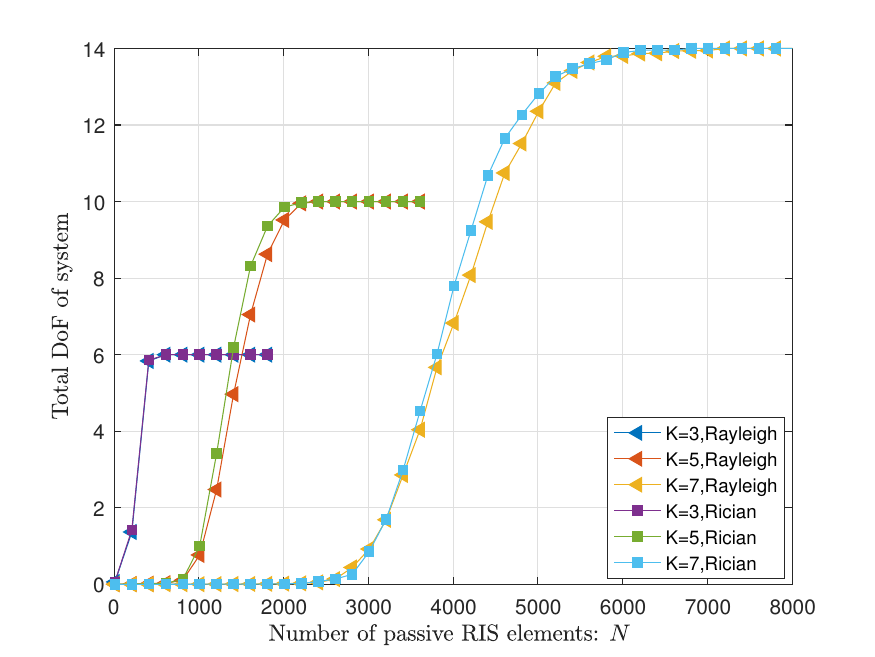}
\label{DoF}
}
\quad
\subfigure[Required REs for full DoF vs. number of users.]{
\includegraphics[width=6cm]{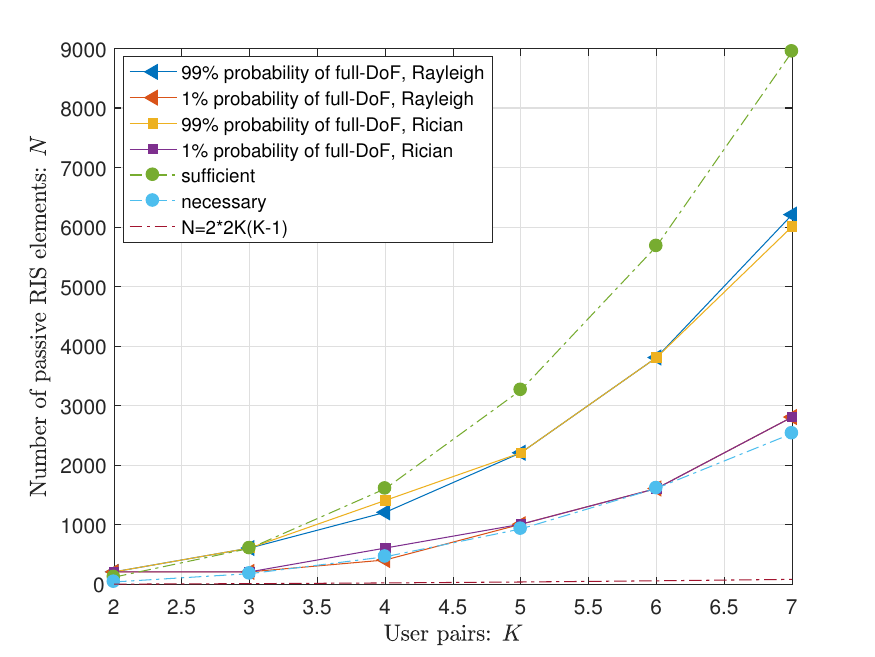}
\label{DoF_N}
}
\caption{Requirement number of REs for full DoF under practical channel settings. }
\label{DoF_RIS}
\end{figure*}

\begin{figure}
	\centering\includegraphics[width=6cm]{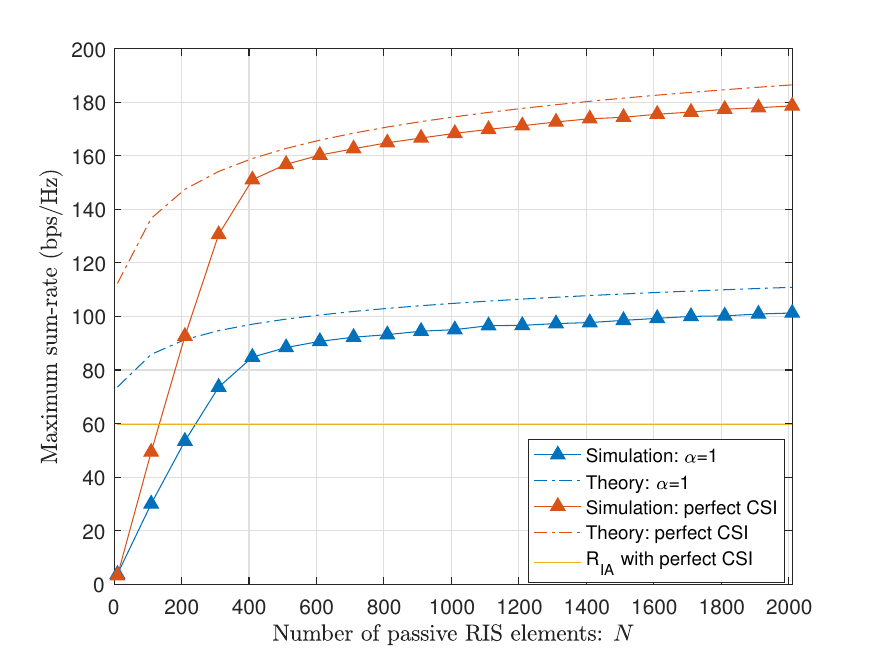}
    \caption{Maximum sum-rate vs. number of REs.}
	\label{upper_bound}
\end{figure}

Furthermore, in Fig.~\ref{DoF_RIS}, we evaluate the required number of passive REs under practical scenarios, where path losses depend on user locations and link distances. Since users on each side are distributed over a fixed region, more users lead to smaller inter-user distances and hence stronger direct links (i.e., smaller path losses). Moreover, each user has a different distance to every other user. Both Rayleigh and Rician channel models are considered.

As shown in Fig.~\ref{DoF}, the full DoF of $2K$ can be achieved once the number of passive REs $N$ exceeds a certain threshold, which is nearly independent of the channel model. Moreover, $N$ increases with the number of users, as more users cause stronger interference, requiring additional REs for suppression.

Fig.~\ref{DoF_N} shows that the actual required REs are much larger than the conventional threshold $N = 2 \times 2K(K - 1)$, which equals twice the number of interference-nulling equations \cite{jiang2022interference}. This discrepancy arises because higher user density leads to stronger same-side direct interference; hence, more passive REs are needed for interference nulling. Therefore, the direct channel strength should be considered in passive RIS design.

In contrast, our theoretical results accurately capture these effects. The “necessary” line corresponds to \eqref{eq13}, where the channel strength ratios ${\eta _1}$ and ${\eta _2}$ are calculated from the average path losses of the cascaded reflective, same-side direct, and cross-side direct links. This result provides a tighter lower bound on the required number of passive REs. While the “sufficient” line in \eqref{eq22} is derived using the average path losses of the same-side and cross-side direct links, with the cascaded reflective link set to its minimum value determined by users at $\left( {50{\text{m}},\pm 15{\text{m}},-20{\text{m}}} \right)$. It can serve as an upper bound.

\subsection{Impact of Imperfect CSI on Sum-Rate and DoF}
\label{scaling_law}
Fig.~\ref{upper_bound} shows the maximum sum-rate versus the number of REs under imperfect CSI. To demonstrate the efficiency of RIS in interference nulling, we compare it with the conventional interference alignment (IA) without RIS. 
By IA, the sum DoF of the bi-directional $K$-user interference network is $K$, and each user still gets half a DoF \cite{cheng2015degrees}. Since the interference-free sum-rate can be expressed as
\begin{equation}
{R_{{\text{no}}{\text{.int}}}} \approx \sum\limits_{i = 1}^{2K} {{\text{Do}}{{\text{F}}_{{\text{full,}}i}} \cdot {{\log }_2}\left( {{\text{SN}}{{\text{R}}_i}} \right)}  = \sum\limits_{i = 1}^{2K} {{{\log }_2}\left( {{\text{SN}}{{\text{R}}_i}} \right)},
\end{equation}
the sum-rate achieved by IA is
\begin{equation}
\begin{gathered}
  {R_{{\text{IA}}}} \approx \sum\limits_{i = 1}^{2K} {{\text{Do}}{{\text{F}}_{{\text{IA,}}i}} \cdot {{\log }_2}\left( {{\text{SN}}{{\text{R}}_i}} \right)}  = \sum\limits_{i = 1}^{2K} {\frac{1}{2}{{\log }_2}\left( {{\text{SN}}{{\text{R}}_i}} \right)}  \hfill \\
   \Rightarrow {R_{{\text{IA}}}} \approx \frac{1}{2}{R_{{\text{no}}{\text{.int}}}} = \frac{1}{2}\sum\limits_{i = 1}^{2K} {{{\log }_2}\left( {1 + \frac{{P{{\left| {{h_{i\left( {\left| {K \pm i} \right|} \right)}}} \right|}^2}}}{{{\sigma ^2}}}} \right)}.  \hfill \\ 
\end{gathered}
\label{eq_IA}
\end{equation}
Moreover, under the MMSE channel estimator, the CSI error decay exponent is $\alpha  = 1$ \cite{marzetta2006fast}. Thus, according to \eqref{eq_s_SNR}, the CSI accuracy factor $\varsigma$ is given by
\begin{equation}
1 - {\varsigma ^2} = \frac{1}{\text{SNR}} \Rightarrow \varsigma  = \sqrt {1-\frac{1}{\text{SNR}}}  \triangleq {\varsigma _{{\text{MMSE}}}}.
\end{equation}

In Fig.~\ref{upper_bound}, the “Simulation” curves are obtained from \eqref{max} using manifold optimization initialized by the alternating projection in \eqref{eq6}. The “Theory” curves represent the upper bound $E\left( {{{\bar R}}} \right)$ in \eqref{eq_E_rate}, where cross-side direct links use average path loss and cascaded reflective links use the maximum value determined by users at $\left( {5{\text{m}},\pm 15{\text{m}},-20{\text{m}}} \right)$. 
It is observed that when the number of REs exceeds a certain threshold (approximately $605$ according to \eqref{num_re}), the growth of the maximum sum-rate aligns well with the theoretical analysis.
Although imperfect CSI (where ${\varsigma_{{\text{MMSE}}}}$ is based on the average received SNR) may reduce the sum-rate due to severe residual interference caused by CSI errors, the sum-rate remains substantially higher than that of conventional interference alignment (i.e., $R_{IA}$ given in \eqref{eq_IA}).
This is because the RIS can still exploit its passive beamforming gain and additional spatial degrees of freedom to boost effective signal power.

Fig.~\ref{sum_rate_SNR} shows the maximum sum-rate and DoF versus transmit power $P$. To ensure the interference with estimated CSI can be completely canceled, we set the number of REs to $N=800$ according to Fig.~\ref{upper_bound} and \eqref{num_re}. The “Theory” curves use the same parameters as in Fig.~\ref{upper_bound}. 
From Fig.~\ref{rate_SNR}, it is observed that in the low-SNR regime (i.e., small $P$), the achievable rates under imperfect CSI increase more rapidly with $P$ than those under perfect CSI, and the corresponding rate slope becomes steeper as the CSI error decay exponent $\alpha$ increases, although the resulting sum rates remain lower than those achieved with perfect CSI. This behavior can be attributed to the residual interference induced by CSI errors, which can be more effectively mitigated at low SNR as $P$ or $\alpha$  increases.
In contrast, when $P$ is large, the rate slopes gradually approach a constant determined by the system DoF. In particular, for large $P$ and $\alpha  = 3$, the achievable rate closely matches that of the perfect-CSI case, 
which indicates that the impact of residual interference becomes negligible.

These observations are further confirmed in Fig.~\ref{DoF_SNR} from a DoF perspective. 
However, when the CSI error is large (i.e., the CSI error decay exponent $\alpha$ is small, such as $\alpha=1, 0.5$), a disparity exists between the simulation and the theoretical prediction, whereas this disparity decreases as the transmit power $P$ increases. This observation is consistent with the analysis in \eqref{eq_DoFup_2} and \eqref{eq_DoF_low_2}. Specifically, the theoretical DoF is characterized as a high-SNR (i.e., $P \to \infty$) asymptotic metric, where the lower-order terms associated with CSI errors are ignored. However, in practice, since $P$ is limited, these neglected terms become non-negligible and cause performance degradation. As $P$ increases, their impact becomes weaker, thereby narrowing the gap between theory and simulation.

\begin{figure*}
\centering
\subfigure[Maximum sum-rate vs. $P$.]{
\includegraphics[width=6cm]{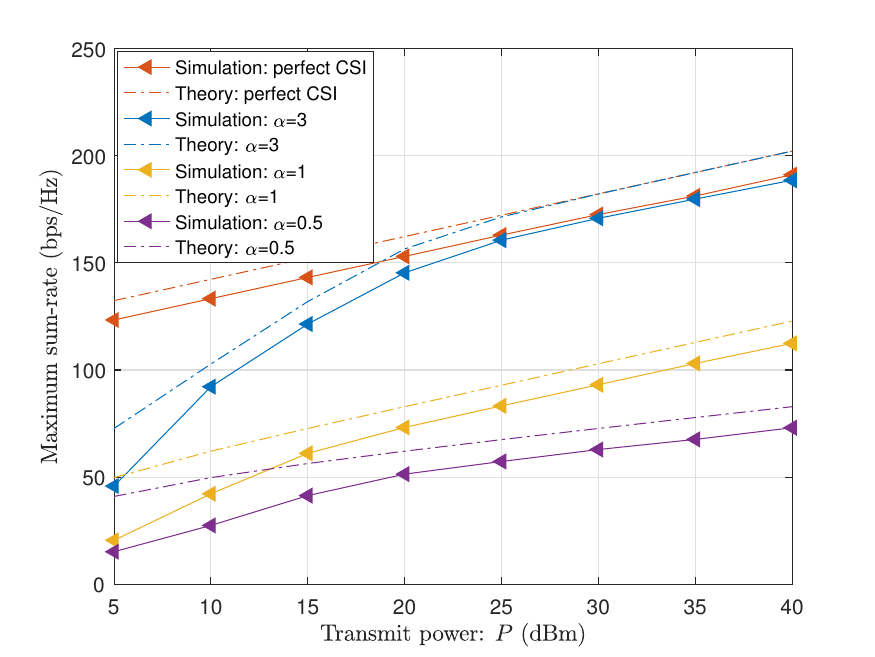}
\label{rate_SNR}
}
\quad
\subfigure[DoF vs. $P$.]{
\includegraphics[width=6cm]{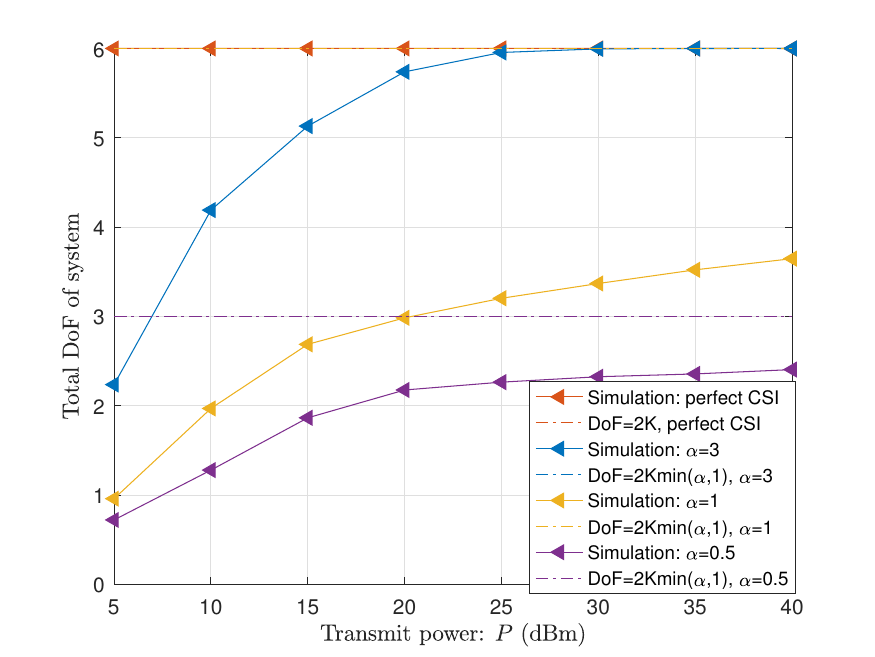}
\label{DoF_SNR}
}
\caption{Maximum sum-rate and DoF versus transmit power $P$ ($K=3$).}
\label{sum_rate_SNR}
\end{figure*}

\section{Conclusion}
\label{Conclusion}
This paper investigates the feasibility of interference-free transmission (in terms of DoF) in the two-way $K$-user interference channel, which is regarded as the most interference-limited network. By leveraging the high dimensionality and randomness of the channel coefficients and exploiting the universal concentration phenomenon in high-dimensional spaces, we show that interference-free transmission is feasible.
Furthermore, we provide the necessary and sufficient conditions on the number of RIS elements, which are of the same order and can be succinctly expressed as $O(K\max({K,\eta}))$, where $\eta$ denotes the strength ratio between the direct and cascaded links.
We also characterize the impact of imperfect CSI on the achievable DoF of the TW-K-IFC. Specifically, a DoF of $2K\min(\alpha,1)$ can be achieved when the CSI error variance scales as $\text{SNR}^{-\alpha}$. This implies that interference-free DoF is preserved when $\alpha \geqslant 1$, which is typically satisfied in practice since $\alpha \approx 1$. Future work includes extending the proposed framework to discrete RIS phase shifts, as well as to multi-RIS systems and scenarios with spatially correlated RIS elements.

\appendix
\subsection{Proof of Theorem \ref{the1}}	
\label{appA}
Since ${{\bar S}_1} = \left\{ {\left. {{\mathbf{x}} \in {\mathbb{C}^{N}}} \right|{\mathbf{x}} = {{\mathbf{x}}_0} + {\mathbf{z}},{\mathbf{z}} \in {\text{null}}\left( {\mathbf{G}} \right)} \right\}$ and
${{\bar S}_2} = \left\{ {\left. {{\mathbf{x}} \in {\mathbb{C}^{N}}} \right|\left| {{x_i}} \right| = \gamma} \right\}$, 
${{\bar S}_2} \subseteq {{\hat S}_2} = \left\{ {\left. {\mathbf{x}} \right|\left| {\mathbf{x}} \right| = \gamma\sqrt N } \right\}$, it follows that if ${{\bar S}_1} \cap {{\bar S}_2} \ne \varnothing$, then ${{\bar S}_1} \cap {{\hat S}_2} \ne \varnothing$. Conversely, if ${{\bar S}_1} \cap {{\hat S}_2} = \varnothing $, then ${{\bar S}_1} \cap {{\bar S}_2} = \varnothing $.
Moreover,
since each $\mathbf{x} \in \bar{S}_1$ can be decomposed as
\begin{equation}
{\mathbf{x}} = {{\mathbf{\tilde x}}} + {{\mathbf{x}}_\parallel },    
\end{equation}
where ${{\mathbf{\tilde x}}} = {\text{Pro}}{{\text{j}}_{{{\left( {{\text{null}}\left( {\mathbf{G}} \right)} \right)}^ \bot }}}\left( {{{\mathbf{x}}_0}} \right)$, ${{\mathbf{x}}_\parallel }$ is a vector lying in the null space of ${\mathbf{G}}$ (i.e., ${{\mathbf{x}}_\parallel } \in {\text{null}}\left( {\mathbf{G}} \right)$), and ${{\mathbf{\tilde x}}} \bot {{\mathbf{x}}_\parallel }$. 
It follows that, if ${{\bar S}_1} \cap {{\hat S}_2} \ne \varnothing$, then there must exist a ${{\mathbf{x}}_\parallel }$ such that
\begin{equation}
{\left| {{{\mathbf{\tilde x}}} + {{\mathbf{x}}_\parallel }} \right|^2} = {\gamma ^2}N \Leftrightarrow {\left| {{\mathbf{\tilde x}}} \right|^2} + {\left| {{{\mathbf{x}}_\parallel }} \right|^2} = {\gamma ^2}N.  
\label{eqa48}
\end{equation}
Thus, a feasible ${{\mathbf{x}}_\parallel }$ exists if and only if 
\begin{equation}
{\left| {{\mathbf{\tilde x}}} \right|^2} \leqslant {\gamma ^2}N \Leftrightarrow \left| {{\text{Pro}}{{\text{j}}_{{{\left( {{\text{null}}\left( {\mathbf{G}} \right)} \right)}^ \bot }}}\left( {{{\mathbf{x}}_0}} \right)} \right| \leqslant \gamma\sqrt N .
\end{equation}
Consequently, if $\left| {{\text{Pro}}{{\text{j}}_{{{\left( {{\text{null}}\left( {\mathbf{G}} \right)} \right)}^ \bot }}}\left( {{{\mathbf{x}}_0}} \right)} \right|>\gamma\sqrt N$, we have ${{\bar S}_1} \cap {{\hat S}_2} = \varnothing $, implying ${{\bar S}_1} \cap {{\bar S}_2} = \varnothing $. This completes the proof.

\subsection{Proof of Theorem \ref{the2}}	
\label{appB}
The proof is obtained by contradiction. We assume that when $d   \leqslant   \frac{r}{{\sqrt 2 }}$, the cross-section ${A'}{O'}{B'}$ of the sphere and ${S_1}$ does not pass through any red point of set ${S_2}$ (i.e., ${S_1} \cap {S_2} = \varnothing $), as shown in Fig. \ref{d3}. 
Since the red points of ${S_2}$ are distributed such that the angle between the vectors from the origin to any two points is $\frac{\pi }{2}$, and $d  \leqslant  \frac{r}{{\sqrt 2 }}$ implies that the full apex angle $\theta$ of the cone (or the central angle of the yellow spherical cap) satisfies $\theta  \geqslant \frac{\pi }{2}$, the cone (or yellow sphere cap) must contain at least one red dot, such as ${{\bf{v}}_m}$.

\begin{figure}[htp] 
	\centering\includegraphics[width=3cm]{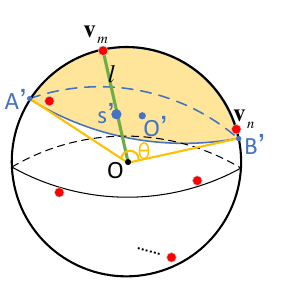}
	\caption{Geometric relationship of ${S_1}$ and ${S_2}$ ($\theta  > \frac{\pi }{2}$).}
	\label{d3}
\end{figure}

Furthermore, since the sphere is intersected by $S_1$ along the circular cross-section ${A'}{O'}{B'}$, and this circular region is entirely contained within $S_1$, the straight line $l$ passing through both the origin $O$ and ${{\bf{v}}_m}$ must intersect $S_1$ at some point $s'$. Let the coordinates of $s'$ be given by $\left( {{v_{s'1}}, \cdots ,{v_{s'N}}} \right)$. Since $s'$ lies along the line connecting $O$ and ${{\bf{v}}_m}$, it
follows that $s'$ can be expressed as a scaled version of ${{\bf{v}}_m}$, i.e.,
\begin{equation}
s' = \mu {{\bf{v}}_m},0 < \mu  < 1.
\end{equation}
Consequently, for each coordinate component, we have
\begin{equation}
\left| {{v_{s'i}}} \right| = \left| {\mu {v_{m,i}}} \right| < \left| {{v_{m,i}}} \right| = 1,
\end{equation}
where ${v_{m,i}}$ denotes the $i$-th element of ${{\bf{v}}_m}$.
It implies that 
\begin{equation}
s' \in \left\{ {{S_1} \cap \left( {{{\bar S}_3}:\left. {\bf{v}} \right|\left| {{v_i}} \right| < 1,\forall i = 1, \cdots ,N} \right)} \right\} \ne \varnothing .
\end{equation}
Therefore, for the $i$-th element of $\bf{v}$ in sets $S_1$ and ${{\bar S}_3}$ (i.e., ${{v_{{S_1},i}}}$ and ${{v_{{{\bar S}_3},i}}}$), we have
\begin{equation}
\begin{aligned}
\left( {{v_{{S_1},i}}:{v_i} = {{\left[ { - {{\left( {{{\bf{A}}^H}} \right)}^ + }{\bf{b}} + {\rm{null}}\left( {{{\bf{A}}^H}} \right)} \right]}_i}} \right) \\
\cap \left( {{v_{{{\bar S}_3},i}}:\left| {{v_i}} \right| < 1} \right) \ne \varnothing .
\end{aligned}
\end{equation}
Since ${v_{{S_1},i}}$ is obtained by translating the null space component ${\left[ {{\rm{null}}\left( {{{\bf{A}}^H}} \right)} \right]_i}$ by ${\left[ {{{\left( {{{\bf{A}}^H}} \right)}^ + }{\bf{b}}} \right]_i}$, as shown in Fig. \ref{fig4}, if ${v_{{S_1},i}} \cap {v_{{{\bar S}_3},i}} \ne \varnothing $, we have ${v_{{S_1},i}} \cap \left( {{v_{{S_2},i}}:\left| {{v_i}} \right| = 1} \right) \ne \varnothing ,\forall i \in \left[ {1,N} \right].$
It is equivalent to ${S_1} \cap {S_2} \ne \varnothing$, which leads to a contradiction. Thus, when $d  \leqslant  \frac{r}{{\sqrt 2 }}$, the cross-section must pass through at least one point of the set $S_2$, which completes the proof.
\begin{figure}[htp] 
	\centering\includegraphics[width=5cm]{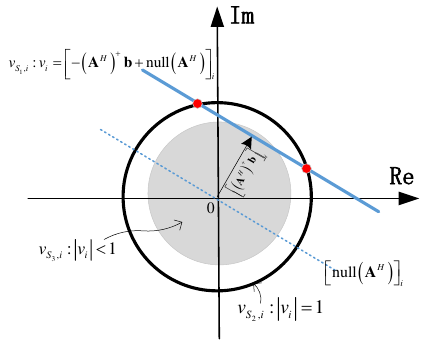}
	\caption{Geometric relationship of the $i$-th element of ${\bf{v}}$ in ${S_1}$, ${S_2}$ and ${{\bar S}_3}$.}
	\label{fig4}
\end{figure}

\textbf{Intuition explanation:}

The intuition behind the proof can be illustrated through a toy example. Specifically, consider only two points of $S_2$ (${{\mathbf{v}}_m}$ and ${{\mathbf{v}}_n}$), whose position vectors are orthogonal, i.e., the angle between them is $\frac{\pi }{2}$. Due to the randomness of ${\mathbf{A}}$ and ${\mathbf{b}}$, the affine subspace $S_1$ can rotate around the sphere on which ${{\mathbf{v}}_m}$ and ${{\mathbf{v}}_n}$ lie. Our goal is to ensure that $S_1$ intersects at least one point of $S_2$ regardless of its rotation. 

As illustrated in Fig. \ref{intui1}, such an intersection is impossible when $d>r$. When $d=r$, an intersection can occur for certain rotations, such as the blue dashed line. However, it cannot be guaranteed for all orientations. In the worst-case, where $S_1$ is not biased toward either point (i.e., blue solid line), $S_1$ does not intersect either point. Moreover, as shown in Fig. \ref{intui2}, the intersection is still not guaranteed even when $d$ is slightly smaller than $r$.

When $d = \frac{r}{{\sqrt 2 }}$, Fig. \ref{intui3} shows that $S_1$ can intersect two points even in the worst case. At first glance, Fig. \ref{intui4} may suggest that this property no longer holds when $d < \frac{r}{{\sqrt 2 }}$. However, in fact, $S_1$ still intersects the points. The apparent contradiction arises because Fig. \ref{intui4} is only a projection of the underlying high-dimensional geometry and therefore does not fully preserve the geometric relationships. In the original space, the element-wise relationship of ${S_1}$, ${{\mathbf{v}}_{s'}}$ and ${{\mathbf{v}}_m}$ can be illustrated by Fig. \ref{fig4}. Specifically,
\begin{equation}
{S_{1,i}} \cap \left\{ {\left. {{v_{s',i}}} \right|\left| {{v_i}} \right| < 1} \right\} \ne \varnothing  \Rightarrow {S_{1,i}} \cap \left\{ {\left. {{v_{m,i}}} \right|\left| {{v_i}} \right| = 1} \right\} \ne \varnothing ,\forall i.
\end{equation}
Since the above relationship holds for every corresponding pair of elements in ${S_1}$, ${{\mathbf{v}}_{s'}}$, and ${{\mathbf{v}}_m}$, we can conclude that ${S_1} \cap {{\mathbf{v}}_m} \ne \varnothing$.

\begin{figure}[htp]
\centering
\subfigure[$r \leqslant d$]{
\includegraphics[width=0.3\linewidth]{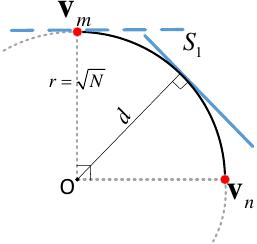}
\label{intui1}
}
\quad
\subfigure[$\frac{r}{{\sqrt 2 }} < d < r$]{
\includegraphics[width=0.32\linewidth]{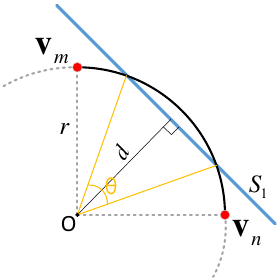}
\label{intui2}
}
\\
\subfigure[$d = r \cdot \cos \left( {\frac{\theta }{2}} \right) = \frac{r}{{\sqrt 2 }}$]{
\includegraphics[width=0.32\linewidth]{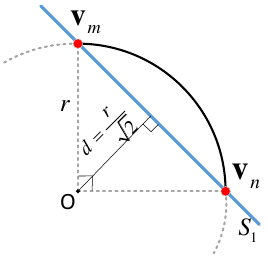}
\label{intui3}
}
\quad
\subfigure[$d < \frac{r}{{\sqrt 2 }}$ or $\theta  > \frac{\pi }{2}$]{
\includegraphics[width=0.3\linewidth]{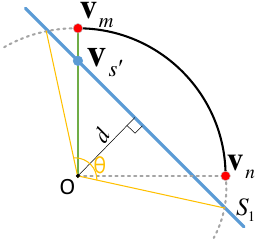}
\label{intui4}
}
\caption{Toy example in a high-dimensional complex space.}
\label{intui}
\end{figure}

\subsection{Proof of the equivalent noise power ${P_{{\text{eq}}}}$ in \eqref{p_eq}}
\label{APP_p_eq}
\begin{equation}
\begin{aligned}
  {P_{{\text{eq}}}} &= E\left( {{{\left| {{n_{{\text{eq1}}}} + {n_{{\text{eq2}}}}} \right|}^2}} \right) \hfill \\
   &= \left( {1 - {\varsigma ^2}} \right)\sum\limits_{j \in [1,2K], \atop j \ne i} {{P_j}E\left( {{{\left| {\Delta {\mathbf{a}}_{ij}^H{\mathbf{v}}} \right|}^2}} \right)}  \hfill \\
   &= \left( {1 - {\varsigma ^2}} \right)\sum\limits_{j \in [1,2K], \atop j \ne i} {{P_j}E\left( {{{\left| {\sum\limits_{n = 1}^N {\Delta a_{ij,n}^*{v_n}} } \right|}^2}} \right)}  \hfill \\
   &= \left( {1 - {\varsigma ^2}} \right)\sum\limits_{j \in [1,2K], \atop j \ne i} {{P_j}N {{L_{Ii}}{L_{Ij}}} }  \hfill \\ 
\end{aligned}
\end{equation}
The last equation follows from the facts that linear combinations of independent Gaussian random variables remain Gaussian and that $\Delta {{\mathbf{a}}_{ij}}$ is rotationally invariant \cite{proakis2001digital,yang2021performance}.
Specifically, since $\Delta {{\mathbf{a}}_{ij}} \sim \mathcal{C}\mathcal{N}\left( {0,\left({{L_{Ii}}{L_{Ij}}}\right){{\mathbf{I}}_N}} \right)$ and ${\mathbf{v}} = {[{e^{j{\beta _1}}}, \ldots ,{e^{j{\beta _N}}}]^T}$, it follows that $\Delta a_{ij,n}^*{e^{j{\beta _n}}} \sim \mathcal{C}\mathcal{N}\left( {0,{{L_{Ii}}{L_{Ij}}} } \right)$ and $\sum\limits_{n = 1}^N {\Delta a_{ij,n}^*{v_n}}  \sim \mathcal{C}\mathcal{N}\left( {0,N {{L_{Ii}}{L_{Ij}}}} \right)$.

\subsection{Proof of the Ergodic ${R}$ Upper Bound}
\label{sumrate_up}
Since the upper bound of ergodic ${R}$ can be expressed as
\begin{equation}
\begin{aligned}
  E\left( {{R}} \right) =& E\left( {\sum\limits_{i = 1}^{2K} {{{\log }_2}\left( {1 + \frac{P}{{{{\tilde P}_{{\text{eq}}}} + {\sigma ^2}}}{{\left| {\varsigma {\mathbf{\tilde a}}_i^H{{\mathbf{v}}_{{\text{opt}}}} + {{\tilde h}_i}} \right|}^2}} \right)} } \right) \hfill \\
  \mathop  \leqslant \limits^{(a)}& E\left( {\sum\limits_{i = 1}^{2K} {{{\log }_2}\left( {1 + \tilde P{{\left( {\varsigma \sum\limits_{n = 1}^N {\left| {{{\tilde a}_{in}}} \right|}  + \left| {{{\tilde h}_i}} \right|} \right)}^2}} \right)} } \right) \hfill \\
  \mathop  \leqslant \limits^{(b)}& 2K{\log _2}\left( {1 + \tilde PE\left( {{X^2}} \right)} \right) \triangleq E\left( {{{\bar R}}} \right) \\ 
\end{aligned}
\label{eq_up_bound},
\end{equation}
where $\tilde P = \frac{P}{{{{\tilde P}_{{\text{eq}}}} + {\sigma ^2}}}$ with ${\tilde P}_{{\text{eq}}}$ defined in \eqref{eq_power}, $X = \varsigma \sum\limits_{n = 1}^N {\left| {{{\tilde a}_{in}}} \right|}  + \left| {{{\tilde h}_i}} \right|$, and ${{\tilde a}_{in}}$ is the $n$-th element of ${{{\mathbf{\tilde a}}}_i}$.
Inequality (a) follows from $\left| v_{{\text{opt}},n} \right| = 1$ and the triangle inequality, and (b) follows directly from Jensen’s inequality.
To calculate $E\left( {{X^2}} \right)$, we first state the following lemma.
\begin{lemma}
\label{rician_model}
Let $h = \sqrt L \left( {\sqrt {\frac{\varepsilon }{{1 + \varepsilon }}} {{\bar h}^{{\text{LoS}}}} + \sqrt {\frac{1}{{1 + \varepsilon }}} {{\tilde h}^{{\text{NLoS}}}}} \right)$, then we have
\begin{equation}
E\left( {\left| h \right|} \right) = \left\{ \begin{gathered}
  \sqrt {\frac{{\pi L}}{4}} \left( {1 + \frac{1}{{16}}{\varepsilon ^2}} \right),{\text{for weak LoS }}\varepsilon  \ll 1 \hfill \\
  \sqrt {\frac{{L\varepsilon }}{{1 + \varepsilon }}} \left( {1 + \frac{1}{{4\varepsilon }}} \right),{\text{for strong LoS }}\varepsilon  \gg 1 \hfill \\ 
\end{gathered}  \right.
.
\label{le}
\end{equation}
\end{lemma}
\begin{proof}
For the weak LoS component, we mean that ${{\sqrt {\frac{\varepsilon }{{1 + \varepsilon }}} } \mathord{\left/
 {\vphantom {{\sqrt {\frac{\varepsilon }{{1 + \varepsilon }}} } {\sqrt {\frac{1}{{1 + \varepsilon }}} }}} \right. \kern-\nulldelimiterspace} {\sqrt {\frac{1}{{1 + \varepsilon }}} }} \ll 1 \Leftrightarrow \varepsilon  \ll 1$. 
The remaining derivation follows directly from \cite{singh2022optimal}. Specifically, the mean of $\left| h \right|$ can be expressed as 
\begin{equation}
E\left( {\left| h \right|} \right) = \sqrt L \sqrt {\frac{1}{{1 + \varepsilon }}} \frac{{\sqrt \pi  }}{2}{L_{{1 \mathord{\left/
 {\vphantom {1 2}} \right.
 \kern-\nulldelimiterspace} 2}}}\left( { - \varepsilon } \right),  
\label{le1}
\end{equation}
where ${L_{{1 \mathord{\left/{\vphantom {1 2}} \right.\kern-\nulldelimiterspace} 2}}}\left(  \cdot  \right)$ is the Laguerre polynomial (see Appendix A in \cite{singh2022optimal}). 
Moreover, according to Section IV-A of \cite{singh2022optimal}, we obtain
\begin{equation}
\sqrt {\frac{1}{{1 + \varepsilon }}} {L_{{1 \mathord{\left/
 {\vphantom {1 2}} \right.
 \kern-\nulldelimiterspace} 2}}}\left( { - \varepsilon } \right) \approx \left\{ {\begin{aligned}
  &{2\sqrt {\frac{\varepsilon }{{\pi \left( {1 + \varepsilon } \right)}}} \left( {1 + \frac{1}{{4\varepsilon }}} \right),{\text{for large }}\varepsilon ,} \\
  &{1 + \frac{1}{{16}}{\varepsilon ^2},{\text{for small }}\varepsilon .} 
\end{aligned}} \right.
\label{le2}
\end{equation}
Substituting \eqref{le2} into \eqref{le1} yields \eqref{le}.
\end{proof}
According to Lemma~\ref{rician_model}, the mean and variance of ${\left| {{{\tilde a}_{in}}} \right|}$ are
\begin{equation}
\begin{aligned}
  E\left( {\left| {{{\tilde a}_{in}}} \right|} \right) &= E\left( {\left| {{{\hat h}_{Ii,n}}} \right|} \right)E\left( {\left| {{{\hat h}_{I\left( {\left| {K \pm i} \right|} \right),n}}} \right|} \right) \hfill \\
   &= \left\{ \begin{gathered}
  \frac{\pi }{4}{\left( {1 + \frac{1}{{16}}{\varepsilon ^2}} \right)^2}\sqrt {{L_{Id}}{L_{Iu}}} ,\varepsilon  \ll 1 \hfill \\
  \frac{\varepsilon }{{1 + \varepsilon }}{\left( {1 + \frac{1}{{4\varepsilon }}} \right)^2}\sqrt {{L_{Id}}{L_{Iu}}} ,\varepsilon  \gg 1 \hfill \\ 
\end{gathered}  \right. \hfill \\ 
\end{aligned}
\label{eq_Ea},
\end{equation}
\begin{equation}
\begin{aligned}
  \operatorname{var} \left( {\left| {{{\tilde a}_{in}}} \right|} \right) &= E\left( {{{\left| {{{\tilde a}_{in}}} \right|}^2}} \right) - {E^2}\left( {\left| {{{\tilde a}_{in}}} \right|} \right) \hfill \\
   &\approx \left\{ \begin{gathered}
  \left( {1 - \frac{{{\pi ^2}}}{{16}}\left( {1 + \frac{{{\varepsilon ^2}}}{4}} \right)} \right){L_{Id}}{L_{Iu}},\varepsilon  \ll 1 \hfill \\
  \left( {\frac{1}{{1 + \varepsilon }}} \right){L_{Id}}{L_{Iu}},\varepsilon  \gg 1 \hfill \\ 
\end{gathered}  \right.
\end{aligned}
\label{eq_vara}.
\end{equation}
Thus, by combining $E\left( {{X^2}} \right) = \operatorname{var} \left( X \right) + {E^2}\left( X \right)$, \eqref{eq_Ea} and \eqref{eq_vara}, we can get
\begin{equation}
\begin{aligned}
  E\left( {{X^2}} \right) &= {\varsigma ^2}N\operatorname{var} \left( {\left| {{{\tilde a}_{in}}} \right|} \right) + \operatorname{var} \left( {\left| {{{\tilde h}_i}} \right|} \right) \\
   &+ {\left( {\varsigma NE\left( {\left| {{{\tilde a}_{in}}} \right|} \right) + E\left( {\left| {{{\tilde h}_i}} \right|} \right)} \right)^2} \hfill \\
  &\mathop  \approx \limits^{N \to \infty } \left\{ \begin{gathered}
  {\left( {1 + \frac{{{\varepsilon ^2}}}{4}} \right)^2}\frac{{{\pi ^2}}}{{16}}{\varsigma ^2}{N^2}{L_{Id}}{L_{Iu}} + \sigma _h^2,\varepsilon  \ll 1 \\
  \frac{\varepsilon}{{1 + \varepsilon }}{\varsigma ^2}{N^2}{L_{Id}}{L_{Iu}} + \sigma _h^2,\varepsilon  \gg 1 \hfill \\ 
\end{gathered}  \right. \hfill \\ 
\end{aligned}
.
\end{equation}
For notational simplicity, we denote $E\left( {{X^2}} \right)$ as
\begin{equation}
E\left( {{X^2}} \right)\mathop  \approx \limits^{N \to \infty } c{\varsigma ^2}{N^2}{L_a} + \sigma _h^2
\label{eq_EX2},
\end{equation}
where ${L_a} = {L_{Id}}{L_{Iu}}$ denotes the strongest (or weakest) channel gain of the cascaded reflective links, and $c$ is a constant defined as
\begin{equation}
c = \left\{ \begin{gathered}
  {\left( {1 + \frac{{{\varepsilon ^2}}}{4}} \right)^2}\frac{{{\pi ^2}}}{{16}},\varepsilon  \ll 1 \hfill \\
  \frac{\varepsilon }{{1 + \varepsilon }},\varepsilon  \gg 1 \hfill \\ 
\end{gathered}  \right.
\label{eq_c}.
\end{equation}
Since for large $N$, we have
\begin{equation}
\begin{aligned}
  \frac{{\partial \left( {\tilde PE\left( {{X^2}} \right)} \right)}}{{\partial {L_a}}} \approx \frac{{c{\varsigma ^2}{\sigma ^2}}}{{{{\left( {1 - {\varsigma ^2}} \right)}^2}{{\left( {2K - 1} \right)}^2}PL_a^2}} > 0, \hfill \\
  \frac{{\partial \left( {\tilde PE\left( {{X^2}} \right)} \right)}}{{\partial \sigma _h^2}} = \frac{P}{{\left( {1 - {\varsigma ^2}} \right)\left( {2K - 1} \right)PN{L_a} + {\sigma ^2}}} > 0. \hfill \\ 
\end{aligned}
\end{equation}
Thus, the upper bound of the ergodic ${R}$ in \eqref{eq_up_bound} is monotonically increasing with the channel gains of the cascaded reflective link (${{L_a}}$) and cross-side direct link (${\sigma _h^2}$). 

Subsequently, by substituting \eqref{eq_power} and \eqref{eq_EX2} into \eqref{eq_up_bound}, we obtain \eqref{eq_E_rate}, which completes the proof.

\subsection{Proof of the DoF Upper Bound}
\label{APP_DoF_up}
With fixed noise power ${{\sigma ^2}}$, ${\text{SNR}} \to \infty$ is equivalent to $P \to \infty $ \cite{joudeh2016sum}. Thus, the upper bound of the DoF can be expressed as
\begin{equation}
{\text{Do}}{{\text{F}}_{{\text{up}}}} = \mathop {\lim }\limits_{{\text{SNR}} \to \infty } \frac{{E\left( {\bar R} \right)}}{{{{\log }_2}\left( {{\text{SNR}}} \right)}} = \mathop {\lim }\limits_{P \to \infty } \frac{{E\left( {\bar R} \right)}}{{{{\log }_2}\left( P \right) + O\left( 1 \right)}}
\label{eq_DoF_up_1}.
\end{equation}
Moreover, according to \eqref{eq_s_SNR}, $\varsigma^2$ can be rewritten as 
\begin{equation}
{\varsigma ^2} = 1 - C{P^{ - \alpha }}
\label{eq_sc},
\end{equation}
where $C$ is a constant with respect to $P$. Thus
\begin{equation}
\begin{aligned}
  &\mathop {\lim }\limits_{P \to \infty } E\left( {\bar R} \right) \\
  &= \mathop {\lim }\limits_{P \to \infty } 2K{\log _2}\left( {1 + \frac{{P\left( {c\left( {1 - C{P^{ - \alpha }}} \right){N^2}{{\bar L}_a} + \bar \sigma _h^2} \right)}}{{C{P^{1 - \alpha }}\left( {2K - 1} \right)N{{\bar L}_a} + {\sigma ^2}}}} \right) \hfill \\
  &= \left\{ \begin{gathered}
  2K\left( {{{\log }_2}\left( {{P^\alpha }} \right) + O\left( 1 \right)} \right){\text{, if }}0 < \alpha  < 1, \hfill \\
  2K\left( {{{\log }_2}\left( P \right) + O\left( 1 \right)} \right){\text{, if }}\alpha  = 1, \hfill \\
  2K\left( {{{\log }_2}\left( P \right) + O\left( 1 \right)} \right){\text{, if }}\alpha  > 1, \hfill \\ 
\end{gathered}  \right. \hfill \\ 
\end{aligned}
\label{eq_DoFup_2}
\end{equation}
substituting \eqref{eq_DoFup_2} into \eqref{eq_DoF_up_1} yields 
${\text{Do}}{{\text{F}}_{{\text{up}}}} = 2K\min \left( {\alpha ,1} \right)$.

\subsection{Proof of the Ergodic ${R}$ Lower Bound}
\label{APP_rate_low}
Since the lower bound of ergodic ${R}$ can be expressed as
\begin{equation}
\begin{aligned}
  E&\left( {{R}} \right) = E\left( {\sum\limits_{i = 1}^{2K} {{{\log }_2}\left( {1 + \frac{P}{{{{\tilde P}_{{\text{eq}}}} + {\sigma ^2}}}{{\left| {\varsigma {\mathbf{\tilde a}}_i^H{{\mathbf{v}}_{{\text{opt}}}} + {{\tilde h}_i}} \right|}^2}} \right)} } \right) \hfill \\
   &\geqslant E\left( {\sum\limits_{i = 1}^{2K} {{{\log }_2}\left( {1 + \tilde P{{\left| {\varsigma {\mathbf{\tilde a}}_i^H{{\mathbf{v}}_{{\text{null}}}} + {{\tilde h}_i}} \right|}^2}} \right)} } \right) \hfill \\
   &\approx 2KE\left( {{{\log }_2}\left( {1 + \tilde P{{\left| {\varsigma {\mathbf{\tilde a}}_i^H{{\mathbf{v}}_{{\text{rand}}}} + {{\tilde h}_i}} \right|}^2}} \right)} \right)  \triangleq E\left( {\underline{R}} \right) \\ 
\end{aligned}
\label{eq_lower_bound},
\end{equation}
where ${{\mathbf{v}}_{\text{null}}}$ denotes the phase shifts designed solely for interference cancellation based on the estimated CSI, and is obtained via the alternating projection algorithm in \eqref{eq6}. Since it depends only on the random interference channels, ${{\mathbf{v}}_{\text{null}}}$ is also random and independent of the signal channels ${{\mathbf{\tilde a}}_i}$. Consequently, after interference nulling, ${{\mathbf{v}}_{\text{null}}}$ can be approximated as a random vector, as indicated by the third term.

Moreover, since ${\mathbf{\tilde a}}_i^H{{\mathbf{v}}_{{\text{rand}}}} = \sum\limits_{n = 1}^N {\tilde a_{i,n}^*{v_{{\text{rand}},n}}}$ and 
\begin{equation}
E\left( {\tilde a_{i,n}^*{v_{{\text{rand}},n}}} \right) = E\left( {\tilde a_{i,n}^*} \right)E\left( {{v_{{\text{rand}},n}}} \right) = 0,
\end{equation}
\begin{equation}
\begin{aligned}
  \operatorname{var} \left( {\tilde a_{i,n}^*{v_{{\text{rand}},n}}} \right) &= E\left( {{{\left| {\tilde a_{i,n}^*{v_{{\text{rand}},n}}} \right|}^2}} \right) - {\left| {E\left( {\tilde a_{i,n}^*{v_{{\text{rand}},n}}} \right)} \right|^2} \hfill \\
   &= E\left( {{{\left| {\tilde a_{i,n}^*} \right|}^2}} \right) = {L_a} \hfill \\ 
\end{aligned}
,
\end{equation}
by the central limit theorem we have $\frac{{{\mathbf{\tilde a}}_i^H{{\mathbf{v}}_{{\text{rand}}}}}}{{\sqrt {N{L_a}} }} \sim \mathcal{C}\mathcal{N}\left( {0,1} \right)$. Thus, the lower bound $E\left( {\underline{R}} \right)$ in \eqref{eq_lower_bound} can be expressed as
\begin{equation}
\begin{aligned}
  &E\left( {\underline{R}} \right) = 2KE\left( {{{\log }_2}\left( {1 + \tilde P{{\left| {\varsigma {\mathbf{\tilde a}}_i^H{{\mathbf{v}}_{{\text{rand}}}} + {{\tilde h}_i}} \right|}^2}} \right)} \right) \hfill \\
   &= 2KE\left( {{{\log }_2}\left( {1 + \tilde P{\varsigma ^2}N{L_a}{{\left| {\frac{{{\mathbf{\tilde a}}_i^H{{\mathbf{v}}_{{\text{rand}}}}}}{{\sqrt {N{L_a}} }} + \frac{{{{\tilde h}_i}}}{{\varsigma \sqrt {N{L_a}} }}} \right|}^2}} \right)} \right) \hfill \\
   &\mathop  = \limits^{\left( a \right)} \frac{{2K}}{{\ln 2}}\int_0^\infty  {\ln \left( {1 + \hat Px} \right)} \frac{1}{\lambda }{e^{ - \frac{x}{\lambda }}}dx \hfill 
   \mathop  = \limits^{\left( b \right)} \frac{{ - 2K}}{{\ln 2}}{e^{\frac{1}{{\hat P\lambda }}}}\text{Ei}\left( { - \frac{1}{{\hat P\lambda }}} \right), \hfill \\ 
\end{aligned}
\label{eq_lower_bound2}
\end{equation}
where $\hat P = \tilde P{\varsigma ^2}N{L_a} = \frac{{P{\varsigma ^2}N{L_a}}}{{{{\tilde P}_{{\text{eq}}}} + {\sigma ^2}}}$ with ${\tilde P}_{{\text{eq}}}$ defined in \eqref{eq_power}.
Equation (a) follows directly from the definition of expected value. Specifically, $z = \frac{{{\mathbf{\tilde a}}_i^H{{\mathbf{v}}_{{\text{rand}}}}}}{{\sqrt {N{L_a}} }} + \frac{{{{\tilde h}_i}}}{{\varsigma \sqrt {N{L_a}} }} \sim \mathcal{C}\mathcal{N}\left( {0,\lambda } \right)$ with variance $\lambda  = 1 + \frac{{\sigma _h^2}}{{{\varsigma ^2}N{L_a}}}$, thus $x = {\left| z \right|^2}$ follows an exponential distribution with probability density function $\text{f}\left( x \right) = \frac{1}{\lambda }{e^{ - \frac{x}{\lambda }}}$. Equation (b) is obtained by applying \cite[eq.(4.337)]{gradshteyn2007table}, 
${\text{Ei}}\left(  \cdot  \right)$ denotes the exponential integral \cite[eq.(8.211)]{gradshteyn2007table}.

\subsection{Proof of the DoF Lower Bound}
\label{DoF_low}
Moreover, based on \eqref{eq_sc}, we have
\begin{equation}
\begin{gathered}
  \mathop {\lim }\limits_{P \to \infty } \hat P\lambda  = \mathop {\lim }\limits_{P \to \infty } \frac{{P{\varsigma ^2}N{L_a}}}{{{{\tilde P}_{{\text{eq}}}} + {\sigma ^2}}}\left( {1 + \frac{{\sigma _h^2}}{{{\varsigma ^2}N{L_a}}}} \right) \hfill \\
   = \mathop {\lim }\limits_{P \to \infty } \frac{{P\left( {\left( {1 - C{P^{ - \alpha }}} \right)N{L_a} + \sigma _h^2} \right)}}{{C{P^{1 - \alpha }}\left( {2K - 1} \right)N{L_a} + {\sigma ^2}}} \hfill \\
   = \left\{ \begin{gathered}
  \frac{{{P^\alpha }\left( {N{L_a} + \sigma _h^2} \right)}}{{C\left( {2K - 1} \right)N{L_a}}} \to \infty {\text{, if }}0 < \alpha  < 1 \hfill \\
  \frac{{P\left( {N{L_a} + \sigma _h^2} \right)}}{{C\left( {2K - 1} \right)N{L_a} + {\sigma ^2}}} \to \infty {\text{, if }}\alpha  = 1{\text{, }} \hfill \\
  \frac{{P\left( {N{L_a} + \sigma _h^2} \right)}}{{{\sigma ^2}}} \to \infty {\text{, if }}\alpha  > 1{\text{, }} \hfill \\ 
\end{gathered}  \right. \hfill \\ 
\end{gathered}
\label{eq_DoF_low_1}
\end{equation}
Therefore, the lower bound $E\left( \underline{R} \right)$ in \eqref{eq_lower_bound2} can be rewritten as
\begin{equation}
\begin{gathered}
  \mathop {\lim }\limits_{P \to \infty } E\left( \underline{R} \right) = \mathop {\lim }\limits_{P \to \infty } \frac{{ - 2K}}{{\ln 2}}{e^{\frac{1}{{\hat P\lambda }}}}{\text{Ei}}\left( { - \frac{1}{{\hat P\lambda }}} \right) \hfill \\
  \mathop  = \limits^{\left( a \right)} \mathop {\lim }\limits_{P \to \infty } \frac{{2K}}{{\ln 2}}\left( { - {C_{{\text{Euler}}}} + \ln \left( {\hat P\lambda } \right) + \frac{1}{{\hat P\lambda }}} \right) \hfill \\
   = \mathop {\lim }\limits_{P \to \infty } 2K{\log _2}\left( {\hat P\lambda } \right) \hfill \\
   = \left\{ \begin{gathered}
  2K\left( {{{\log }_2}\left( {{P^\alpha }} \right) + O\left( 1 \right)} \right){\text{, if }}0 < \alpha  < 1, \hfill \\
  2K\left( {{{\log }_2}\left( P \right) + O\left( 1 \right)} \right){\text{, if }}\alpha  = 1, \hfill \\
  2K\left( {{{\log }_2}\left( P \right) + O\left( 1 \right)} \right){\text{, if }}\alpha  > 1, \hfill \\ 
\end{gathered}  \right. \hfill \\ 
\end{gathered}
\label{eq_DoF_low_2}
\end{equation}
where equation (a) is obtained by applying \cite[eq.(8.214)]{gradshteyn2007table}, and ${C_{{\text{Euler}}}} \approx 0.577$ is the Euler’s constant. Thus, similar to \eqref{eq_DoF_up_1}, we have
${\text{Do}}{{\text{F}}_{{\text{low}}}} =\mathop {\lim }\limits_{P \to \infty } \frac{{E\left( {\underline{R} } \right)}}{{{{\log }_2}\left( P \right) + O\left( 1 \right)}} = 2K\min \left( {\alpha ,1} \right)$.

\bibliographystyle{IEEEtran} 
\bibliography{ref}

\end{document}